\documentclass[12pt]{amsart}

\usepackage{amssymb,amsthm,amsmath,amscd, graphicx}

\usepackage{epstopdf}
 \DeclareMathOperator{\diam}{diam}
 \def\RR{{\mathbb R}}
\def\OL{\overline}
\def \UL {\underline}
\def\({\left(}
\def\){\right)}
\def\[{\left[}
\def\]{\right]}
\def \NN{{\mathbb N}}

\def \W{\widetilde}
\newtheorem{theorem}{Theorem}[section]
\newtheorem{lemma}[theorem]{Lemma}

\newtheorem {proposition}[theorem]{Proposition}
\newtheorem {corollary}[theorem]{Corollary}

\begin{document}
\date{\today}

\title[Low Complexity Energy Methods]{
LOW COMPLEXITY METHODS FOR DISCRETIZING MANIFOLDS VIA  RIESZ ENERGY MINIMIZATION  }

\author [S.V. Borodachov]{S.V. Borodachov}
\address {Towson University, Towson, MD, 21252, USA}
\address{Center for Constructive Approximation, Department of Mathematics, \hspace*{.1in}
Vanderbilt University,
Nashville, TN 37240, USA  } 

\email {sborodachov@towson.edu}
\author[D.P. Hardin]{D.P. Hardin}
 \email{doug.hardin@vanderbilt.edu}
\author[E. B. Saff]{E.B. Saff}
\email{edward.b.saff@vanderbilt.edu}

\thanks{The research of the authors was supported, in part,
by the U. S. National Science Foundation under grant DMS-1109266.  
}


\begin{abstract}

Let $A$ be a compact $d$-rectifiable set embedded in Euclidean space $\RR^p$, $d\le p$.  
For a given continuous distribution $\sigma(x)$ with respect to   $d$-dimensional
Hausdorff measure on $A$, our earlier results provided a method for generating $N$-point 
configurations on $A$
that
  have asymptotic distribution $\sigma (x)$
as $N\to \infty$; moreover such configurations are ``quasi-uniform'' in the sense that the ratio of the covering radius to the separation distance is bounded independent of $N$.  The method is based upon minimizing the energy of $N$ particles constrained to $A$ interacting via a weighted power law potential $w(x,y)|x-y|^{-s}$, where $s>d$ is a fixed parameter and 
$w(x,y)=\left(\sigma(x)\sigma(y)\right)^{-({s}/{2d})}$.

Here we show that one can  generate points on $A$ with the above mentioned properties keeping in the energy sums only those  pairs  of points that are located at a distance of at most $r_N=C_N N^{-1/d}$ from each other,  with $C_N$ being a positive sequence tending to infinity  arbitrarily slowly. To do this we minimize the energy with respect to a varying truncated weight $v_N(x,y)=\Phi\(\left|x-y\right|/r_N\)w(x,y)$, where $\Phi:(0,\infty)\to [0,\infty)$ is a bounded function with $\Phi(t)=0$, $t\geq 1$, and $\lim_{t\to 0^+}\Phi(t)=1$. This  reduces, under appropriate assumptions,  the complexity of generating $N$ point `low energy' discretizations to order $N C_N^d$ computations.   
\end{abstract}
\maketitle

\section {Introduction}

Points on a compact set $A$ that minimize certain energy functions often have desirable properties that reflect 
special features of $A$.  For $A=S^2$, the unit sphere in $\RR^3$, the determination of minimal Coulomb
energy points is the classic problem of Thomson \cite{Mel77, Bow00}.  Other energy functions on higher dimensional 
spheres give rise to equilibrium points that are useful for a variety of applications including coding theory \cite{ConSlo99}, cubature formulas \cite{SloWom}, and  
the generation of finite normalized  tight frames \cite{BenFic03}.  In this paper, we shall consider a generalized Thomson problem, namely minimal energy points for weighted Riesz potentials on rectifiable sets (where the weight varies as the cardinality of the configuration grows). Energy problems with varying weights arise, in particular, in physical problems involving  potentials that are not scale invariant.    

Our focus is on the hypersingular
case when short range interaction between points is the dominant effect.  Such energy functions are not 
treatable with classical potential theoretic methods, and so require different techniques of analysis.  

Let $A$ be a compact  set in $\RR^{p}$ whose $d$-dimensional Hausdorff measure\footnote{For integer $d$,  we normalize Hausdorff measure on $\RR^p$ so that $\mathcal{H}_d(U)=1$ if $U$ is a $d$-dimensional unit cube embedded in   $\RR^p$.}, $\mathcal{H}_d(A)$,  is finite and positive.
For  a collection of $N\geq 2$ distinct points
$\omega_N:=\{x_1,\ldots ,x_N\}\subset A$, a non-negative weight function $w$ on $A\times A$ (we shall specify additional conditions on $w$ shortly),
and $s>0$,
 the  {\em weighted Riesz $s$-energy of $\omega_N$} is defined by
$$
E^w_s(\omega_N):=\sum_{1\leq i\neq j\leq N}{\frac
{w(x_i,x_j)}{\left |x_i-x_j\right|^s}}=\sum_{i=1}^N\sum_{ {j=1} \atop { j\neq i}}^N\frac{w(x_i,x_j)}{\left|x_i-x_j\right|^s},
$$
while  the {\em $N$-point
weighted Riesz $s$-energy of $A$} is defined by
\begin{equation} \label{c2'}\mathcal E^w_s(A,N):=\inf \{E^w_s(\omega_N) : \omega_N\subset A ,
 \# \omega_N =N\},
\end{equation} where $\# X$ denotes
the cardinality of a set $X$.  If $v(x,y)=(w(x,y)+w(y,x))/2$, then $E_s^v(\omega_N)=E_s^w(\omega_N)$ for any $N$-point configuration $\omega_N\subset A$, and so, without loss of generality, we assume that $w$ is symmetric; i.e.,
 $w(x,y)=w(y,x)$ for $x,y\in A$.

We call $w:A\times A\to [0,\infty]$ a  {\em CPD-weight function on $A\times A$} if
\begin{enumerate}
\item[(a)] $w$ is continuous (as a function on $A\times A$) at $\mathcal{H}_d$-almost every point of
the diagonal $D(A):=\{(x,x) : x\in A\}$,
\item[(b)] there is some neighborhood $G$ of $D(A)$ (relative to $A\times A$) such that $\inf_G w>0$,  and
\item[(c)] $w$ is bounded on every closed subset $B\subset A\times A\setminus D(A)$.
\end{enumerate}
Here CPD stands for (almost everywhere) continuous and positive on the diagonal.
In particular, conditions (a), (b), and (c) hold if $w$ is bounded on $A\times A$ and continuous and positive
at every point  of the diagonal  $D(A)$ (where
continuity at a diagonal point $(x_0,x_0)$ is meant in the sense of limits taken on
$A\times A$).    We mention that if a CPD-weight $w$ is also lower semi-continuous on $A\times A$, then  the infimum in \eqref{c2'} will be attained.

If $w\equiv 1$ on $A\times A$ (which we refer to as the {\em unweighted} case),  we write
$E_s(\omega_N)$ and $\mathcal E_s(A,N)$ for $ E^w_s(\omega_N)$ and $\mathcal E^w_s(A,N)$, respectively.   For the trivial cases $N=0$ or $1$ we put
$E_s(\omega_N)=\mathcal E_s(A,N)=E^w_s(\omega_N)=\mathcal
E^w_s(A,N)=0$.

In previous works, the authors of this paper have investigated asymptotics as $N\to \infty$ for a  fixed weight $w$ for the energy $\mathcal
E^w_s(A,N)$ as well as for the optimal configurations  that achieve the minimum energy.  Our focus in this article is a generalization that allows the weight $w$ to vary with $N$.  A primary motivation for this generalization is to lower the complexity of energy computations that typically are of order $N^2$ by incorporating a ``cut-off'' function into the weight that depends on $N$.

Before stating our main results we provide some needed notation and review some relevant prior work.

A set $A\subset \RR^{p}$ is called {\it $d$-rectifiable} if $A=\phi (K)$, where $K\subset \RR^d$ is a bounded set and $\phi:K \to \RR^{p}$ is a Lipschitz mapping. A set $A\subset \RR^{p}$ is called {\em $(\mathcal H_d,d)$-rectifiable} if $\mathcal H_d(A)<\infty$ and $A$ is a union of at most a countable collection of $d$-rectifiable sets and a set of $\mathcal H_d$-measure zero.

A sequence of Borel probability measures $\{\mu_N\}$ supported on a compact set $A$ in $\RR^{p}$ is said to {\em converge in the weak* sense} to a Borel probability measure $\mu$ (supported on $A$), if for every Borel subset $B$ of $A$ whose relative boundary $\partial _A B$ with respect to $A$ has $\mu$-measure zero, we have 
$$
\lim_{N\to\infty}{\mu_N(B)}=\mu (B).
$$  
In this case we write 
$
\mu_N\stackrel{*}{\longrightarrow} \mu \text{ as } N\to\infty.
$

Let $\mathcal L_m$ denote the Lebesgue measure in $\RR^m$  and let
$$
K(\epsilon):=\{x\in \RR^{p} : {\rm dist}(x,K)<\epsilon\}
$$ 
denote the $\epsilon$-neighborhood of the set $K$ in $\RR^{p}$. The {\em upper} and the {\em lower $d$-dimensional Minkowski content} of the set $K$ are defined by
$$
\OL {\mathcal M}_d(K):=\limsup_{\epsilon\to 0^+}{\frac {\mathcal L_{p}(K(\epsilon))}{\beta_{p-d}\epsilon ^{p-d}}}
$$
and
$$
\UL {\mathcal M}_d(K):=\liminf_{\epsilon\to 0^+}{\frac {\mathcal L_{p}(K(\epsilon))}{\beta_{p-d}\epsilon ^{p-d}}}
$$
respectively, where $\beta_{m}$ is the Lebesgue measure of the unit ball in $\RR^m$, $m\in \NN$, and $\beta_0:=1$. If the limit
$$
{\mathcal M}_d(K):=\lim_{\epsilon\to 0^+}{\frac {\mathcal L_{p}(K(\epsilon))}{\beta_{p-d}\epsilon ^{p-d}}}
$$
exists, it is called the {\em $d$-dimensional Minkowski  content} of the set $K$. We also let  $\delta_x$  denote the unit point mass at  $x\in \RR^{p}$.

For $s>0$ and a CPD-weight $w$ on $A$, we say that a sequence $\{\omega_N\}_{N=2}^\infty$  of $N$-point configurations on $A$   is {\em asymptotically $(w,s)$-energy minimizing}  if
$$
\lim_{N\to\infty}{\frac {E^w_s(\omega_N)}{\mathcal E^w_s(A,N)}}=1.
$$

In the unweighted case ($w\equiv 1$) the asymptotic behavior of the minimal energy and the weak* limit distribution of energy minimizing configurations are known for wide classes of sets as stated in the following theorem.   
\begin {theorem}\label {packing}
Let $s > d$ and $p\geq d$, where $d$ and $p$ are integers. For every
infinite compact $(\mathcal H_d, d)$-rectifiable set $A$ in $\RR^{p}$ with $\mathcal M_d(A) = \mathcal H_d(A)$, we have
$$
\lim_{N\to \infty}{\frac {\mathcal E_{s}(A,N)}{N^{1+s/d}}}=\frac { C_{s,d}}{\mathcal H_d(A)^{s/d}},
$$
where $C_{s,d}$ is a positive and finite constant independent of $A$.

Moreover, if $A$ is $d$-rectifiable with $\mathcal H_d(A) > 0$, then any sequence $\{\omega_N^\ast\}_{N=2}^{\infty}$ of asymptotically $s$-energy
minimizing configurations on $A$ such that $\# \omega_N^\ast = N$ is asymptotically uniformly
distributed on $A$ with respect to $\mathcal H_d$, i.e.
\begin {equation}\label {q}
\frac {1}{N}\sum_{x\in \omega_N}{\delta_{x}}\stackrel{*}{\longrightarrow} \frac {\mathcal H_d(\cdot\cap A)}{\mathcal H_d(A)},\ \ \ N\to\infty.
\end {equation}
\end {theorem}
This result was proved for the case that  $A$ is a finite union of rectifiable Jordan arcs in \cite [Theorems 3.2 and 3.4]{MMRS},  a $d$-dimensional rectifiable manifold in \cite [Theorem 2.4]{HarSaf05},  a $d$-rectifiable closed set in \cite [Theorems 1 and 2]{BorHarSaf08}, and, in the form presented above, in \cite [Theorems 1 and 2 and related remarks]{BorHarSaf08}.

We remark that the constant $C_{s,1}=2\zeta(s)$ for $s>1$,  where $\zeta(s)$ denotes the classical Riemann zeta function.  
For  other values $d$, this constant   is not yet known. However, for certain values of $d$, specifically $d=2, 4, 8$ and 24, it is conjectured (cf. \cite{BrauHarSaff}) that     
$\displaystyle C_{s,d}=|\Lambda_d|^{s/d}\zeta_{\Lambda_d} (s)$ for $s>d$,
where  $\zeta_{\Lambda_d}$  denotes the {\em Epstein zeta function} for     the hexagonal, $D_4$, $E_8$, and  Leech lattices, respectively and
$|\Lambda_d|$ denotes the co-volume of $\Lambda_d$.

Given a CPD-weight $w$ on $A\times A$,  define for any Borel set $B\subset A$, 
$$
\mathcal H_d^{s,w}(B):=\int_{B}{w(x,x)^{-d/s}{\rm d}\mathcal H_d(x)},\ \ \ s>d.
$$ 
If $0<\mathcal H_d(A)<\infty$, the corresponding probability measure on $A$ is 
\begin{equation}\label{hdsw}
h^{s,w}_d(B):=\frac {\mathcal H^{s,w}_d(B)}{\mathcal H^{s,w}_d(A)}.
\end{equation}
In the case of weighted energy the following asymptotic result is known, see \cite [Theorem 2]{BorHarSaf08}.
\begin {theorem}\label {Thknown}
Let $A\subset \RR^{p}$ be an infinite closed $d$-rectifiable set. Suppose $s>d$ and that $w$ is a $CPD$-weight function on $A\times A$. Then
$$
\lim_{N\to \infty}{\frac {\mathcal E^{w}_{s}(A,N)}{N^{1+s/d}}}=\frac {C_{s,d}}{\[\mathcal H_d^{s,w}(A)\]^{s/d}},
$$
where the constant $C_{s,d}$ is as in Theorem~\ref{packing}. Furthermore, if $\mathcal H_d(A)~>~0$, any asymptotically $(w,s)$-energy minimizing sequence of $N$-point configurations  on $A$ is uniformly distributed with respect to the probability measure $h^{s,w}_d$  defined in \eqref{hdsw}, as $N\to\infty$.
\end {theorem}

One application of the above theorem is to generate points on a rectifiable set that have a specified limiting distribution with respect to Hausdorff measure on the set.  More precisely, 
if $A$ is as in Theorem~\ref{Thknown} and    $\sigma$ is a probability density on $A$ that is continuous almost everywhere with respect to $\mathcal H_d$  and is bounded above and below by positive constants, then for  fixed  $s>d$ and $w:A\times A\to [0,\infty)$    given by
  \begin{equation} \label{wrhodef}
w(x,y):=(\sigma (x)\sigma (y))^{-s/2d},\end{equation}   a sequence of normalized counting measures associated with   $N$-point $(w,s)$-energy minimizing configurations on $A$ converges
  weak*  (as $N\to \infty$) to $\sigma (\cdot)\,  \mathrm{d}\mathcal{H}_d(\cdot)$ (see also,  \cite [Corollary 2]{BorHarSaf08}).

 The outline of the paper is as follows.  In the next section we state  our main results.  In Section~\ref{section:Comp},
we provide complexity estimates for generating minimum weighted energy points that involve a  
cut-off function, and we illustrate the generation method with two examples---one for the sphere and another for
a 3-dimensional spherical shell. Section~\ref{proof2.1} is devoted to the proof of Theorem~\ref{rth1},  while Sections~\ref{proof2.2} and  \ref{proof2.2b}
are devoted to the proofs of Theorems~\ref{Th1} and \ref{C2}.  The proofs of Theorems~\ref{S1} and \ref{S2} are given in Section~\ref{sepsec}
and the complexity assertions from Section~\ref{section:Comp} are justified in Section~\ref{CompProofs}.

\section {Main results}\label {S2.1}

The main purpose of this paper is to present an efficient method for generating a large number of points on a manifold that are well-separated and approximate a given distribution. 
The low complexity of our method is accomplished by performing significantly fewer operations when computing energy sums and gradients.

We begin by stating the following  result extending Theorem \ref {Thknown} to the wider class of $(\mathcal H_d,d)$-rectifiable sets whose Minkowski content of dimension $d$ coincides with the $d$-dimensional Hausdorff measure.   We note  that this result also extends relation (\ref {q}) of Theorem \ref {packing} to this class of sets.  The proof of this result will appear in Section~\ref{proof2.1}.
\begin {theorem}\label{rth1}
Let $A\subset \RR^{p}$ be an infinite compact $(\mathcal H_d,d)$-rectifiable set with $\mathcal H_d(A)=\mathcal M_d(A)$ and suppose that $w$ is a CPD-weight
function on $A\times A$.  If $s>d$, then
\begin {equation}\label {w3j}
\lim_{N\to \infty}{\frac{\mathcal
E^w_s(A,N)}{N^{1+s/d}}}=\frac{C_{s,d}}{\left[\mathcal{H}_d^{s,w}(A)\right]^{s/d}},
\end {equation}
where the constant $C_{s,d}$ is as in Theorem~\ref{packing}.

Furthermore, if ${\mathcal H_{d}(A)~>~0}$, any sequence $\W\omega_N=\{x^N_1,\ldots,x^N_N\}$   of asymptotically
($w,s$)-energy minimizing configurations on $A$  
    is uniformly distributed   with respect to the probability measure $h_d^{s,w}$ as $N\to \infty$.\end{theorem}

The following theorem,  one of the main results of this paper, concerns asymptotic results in the case when the weight function includes a ``cut-off'' function depending on $N$. 
Given a sequence of non-negative weights $\mathbf{v}=\{v_N\}$ on $(A\times A)\setminus D(A)$, we say that a sequence of $N$-point configurations $\{\omega_N\}$ on $A$  is   {\em asymptotically $(\mathbf{v},s)$-energy minimizing}  if
$$
\lim_{N\to\infty}{\frac {E^{v_N}_s(\omega_N)}{\mathcal E^{v_N}_s(A,N)}}=1.
$$ 

 \medskip

\begin {theorem}\label {Th1}
Let   $A \subset\RR^p$  be a compact $(\mathcal H_d,d)$-rectifiable set with $\mathcal H_d(A)=\mathcal M_d(A)>0$ and let $w$ be a CPD-weight
function on $A\times A$.  Suppose $\Phi$ is a non-negative, bounded function on $(0,\infty)$ such that $\lim_{t\to 0^+}{\Phi(t)}=1$  and   $\{r_N\}_{N\in \NN}$ is a sequence of positive numbers such that 
\begin{equation}\label {r_N}
\lim_{N\to\infty}r_NN^{1/d}=\infty.  
\end{equation}
For $N\in \NN$, let $\mathbf{v}=\{v_N\}_{N\in\NN}$ denote the sequence of weights
\begin {equation}\label {v}
v_N(x,y):=\Phi\(\frac {\left|x-y\right|}{r_N}\)w(x,y),\ \ \ x,y\in A,\ \ x\neq y.
\end {equation}
If $s>d$, then
\begin{equation}\label {Thm2.2c}
\lim_{N\to \infty}{\frac {\mathcal E^{v_N}_{s}(A,N)}{N^{1+s/d}}}=\frac {C_{s,d}}{\[\mathcal H_d^{s,w}(A)\]^{s/d}},
\end {equation}
where the constant $C_{s,d}$ is as in Theorem~\ref{packing}.
Furthermore, any sequence of  
  asymptotically
$(\mathbf{v},s)$-energy minimizing $N$-point configurations on $A$ is uniformly distributed with respect to the probability measure $h_d^{s,w}$,
as $N\to \infty$. 
\end {theorem}

The proof of Theorem~\ref{Th1} is given in Sections~\ref{proof2.2} and \ref{proof2.2b}.

\medskip

\noindent{\bf Remark:}  Note that {\em any} compact set $A\subset \RR^d$ (i.e., $p=d$) is automatically a $(\mathcal H_d,d)$-rectifiable set with  $\mathcal H_d(A)=\mathcal M_d(A)$ and so the conclusions of Theorem~\ref{Th1} hold for any compact $A\subset \RR^d$  such that $\mathcal H_d(A)>0$.  The same is true for any $A\subset \RR^p$ that is compact and $d$-rectifiable. \\

Theorem \ref {Th1} implies, in particular, that the  minimal $(v_N,s)$-energy has the same asymptotic dominant term on a wide class of compact rectifiable sets as the minimal $(w,s)$-energy (for $s>d$).  

We note that  if $\Phi(t)=0$ for $t>1$, then the energy sum $
E^{v_N}_s(\omega_N)$ for this cutoff function simplifies since it only involves  pairs of points from $\omega_N$ that are no further than $r_N$ apart.
In the next section, we discuss the complexity of computing such sums in more detail. \\

Theorem \ref {Th1} is a consequence of a more general result, which we present next. It provides general conditions under which one can find the asymptotic behavior of the minimal weighted energy sum where the weight varies with $N$. In view of condition (b) of the definition of a CPD-weight, there is a number $\kappa>0$ such that $w(x,y)>0$, whenever $x,y\in A$ and $\left|x-y\right|<\kappa$.
Given a  non-negative function $v(x,y)$ on $A\times A\setminus D(A)$, for every $\delta\in (0,\kappa)$, define
\begin{equation}\label{Iwv}
I^w(v,\delta)=\inf\{\frac {v(x,y)}{w(x,y)} : (x,y)\in A\times A,\ 0<\left|x-y\right|\leq \delta\}
\end {equation}
and let
\begin{equation}
S^w(v,\delta)=\sup\{\frac {v(x,y)}{w(x,y)} : (x,y)\in A\times A,\ 0<\left|x-y\right|\leq \delta\}.
\end {equation}\label{Swv}

\begin {theorem}\label {C2}
Let $s>d$, $A\subset \RR^{p}$  be a compact $(\mathcal H_d,d)$-rectifiable set with $\mathcal H_d(A)=\mathcal M_d(A)>0$, $w$ be a CPD-weight
function on $A\times A$, and $\mathbf{v}=\{v_N\}_{N\in\NN}$ be a sequence of non-negative functions on $A\times A\setminus D(A)$ such that for some constant $M>0$, 
\begin{equation}\label{1"}
v_N(x,y)\leq Mw(x,y),\ \ \ (x,y)\in A\times A\setminus D(A), \  N\in\NN,
\end{equation}
and
\begin {equation}\label {1'}
\lim_{N\to\infty}{I^w(v_N,aN^{-1/d})}=\lim_{N\to\infty}{S^w(v_N,aN^{-1/d})}=1
\end {equation}
for every positive constant $a$.
Then
\begin {equation}\label {asympt''}
\lim_{N\to \infty}{\frac {\mathcal E^{v_N}_s(A,N)}{N^{1+s/d}}}=\frac {C_{s,d}}{\[\mathcal H^{s,w}_d(A)\]^{s/d}},
\end {equation}
where the constant $C_{s,d}$ is as in Theorem~\ref{packing}. 

Furthermore, any sequence of asymptotically $(\mathbf{v} , s)$-energy minimizing $N$-point configurations on $A$ is uniformly distributed with respect to the probability measure $h^{s,w}_d$ as $N\to \infty$.
\end {theorem}

The proof Theorem~\ref{C2} is given in Section~\ref{proof2.2b}.
\medskip

We next find conditions that guarantee that a sequence of $\(v_N, s\)$-energy minimizing $N$-point configurations is {\em quasi-uniform}, that is, the ratios of the covering radius\footnote{The covering radius of a configuration (relative to a set $A$) is also referred to as the {\em fill radius}  or the {\em mesh-norm} of the configuration.} to the separation distance of the configurations stay bounded as   $N\to\infty$.  For a point configuration $X$ in $\RR^{p}$, we define its {\em separation distance} by
\begin{equation}\label{sepdef}
\delta(X):= \inf_{x,y\in X \atop x\neq y}{\left|x-y\right|},
\end{equation}
and its  {\em covering radius relative to a set $A$ in $\RR^p$} by
\begin{equation}\label{covdef}
\rho(X,A):=\sup_{y\in A}\inf_{x\in X} {\left|x-y\right|}.
\end{equation}
We shall establish quasi-uniformity of $\(v_N, s\)$-energy minimizing $N$-point configurations $\omega_N^s$ in $A$ by showing that both
$\delta(\omega_N^s)$ and $\rho(\omega_N^s, A)$ are of order $N^{-1/d}$. 

\begin {theorem}\label {S1}
Let $s>d$,  $A\subset \RR^{p}$ be a compact  set with $ \mathcal H_d (A)>0$, and $\{v_N\}$ be a uniformly bounded sequence of non-negative lower semi-continuous functions on $A\times A$ such that for $N$ sufficiently large, there holds
\begin{equation}\label{S1cond}
  v_N(x,y)>\alpha_0, \qquad (x,y)\in A\times A,\ 0<\left|x-y\right|\leq a_0N^{-1/d},
\end{equation} for some    positive constants $a_0$ and $\alpha_0$. Then for every sequence $\{\omega^s_N\}$ of $N$-point  $(v_N,s)$-energy minimizing configurations on $A$, there holds
\begin{equation}\label{S1eq1}
\liminf_{N\to\infty}\delta(\omega^s_N)N^{1/d}>0.
\end{equation}
\end {theorem}

We note that Theorem~\ref{S1}  holds under the assumptions on $s$, $A$,  and $\{v_N\}$ in Theorem~\ref{C2} provided that 
the $v_N$'s are uniformly bounded and lower semi-continuous on $A\times A$.  

For the next result concerning the covering radius, we recall the notion of a {\em   $d$-regular} set.  A compact set $\tilde A\subset \RR^p$ is said to be {\em   $d$-regular} if
there exists   a finite positive Borel measure $\mu$ supported on $\tilde A$ that is both upper and lower $d$-regular, that is, 
there are   positive constants $c_0, C_0$ such that
\begin{equation}\label{mureg}
c_0^{-1}r^d\le  \mu(B(x,r))\le C_0r^d, \qquad (x\in \tilde A, \, 0<r<\diam \tilde A),
\end{equation}
where   $B(x,r)$ denotes the open ball in $\RR^{p}$ centered at $x$ of radius $r>0$.

\begin {theorem}\label {S2}
Assume that $s$, $A$,  and $\{v_N\}$ are as in Theorem~\ref{C2}.  In addition, assume that the $v_N$'s are uniformly bounded and lower semi-continuous on $A\times A$, and that $A$ is a subset of a   $d$-regular set $\tilde A\subset\RR^p$.  Then for every sequence $\{\omega^s_N\}$ of $N$-point configurations on $A$ such that $\omega^s_N$ minimizes the $(v_N,s)$-energy, $N\in \NN$, there holds
\begin{equation}\label{limsuprho}
\limsup_{N\to\infty}\rho(\omega^s_N,A)N^{1/d}<\infty.
\end{equation}
\end {theorem}

The proofs of Theorems~\ref{S1} and \ref{S2} are given in Section~\ref{sepsec}.

\bigskip

In applications with a non-uniform limiting density, it can be useful to allow the `cutoff' radius $r_N=r_N(x,y)$ in \eqref{r_N} to depend on $(x,y)\in A\times A$.  The following immediate corollary of Theorems~\ref{C2}, \ref{S1} and \ref{S2} addresses this case. 

\begin{corollary} \label{CorC2} Let $s>d$, $A\subset \RR^{p}$   be a compact $(\mathcal H_d,d)$-rectifiable set with $\mathcal H_d(A)=\mathcal M_d(A)>0$ and let $w$ be a CPD-weight
function on $A\times A$. 
Suppose $r_N:A\times A\to (0,\infty)$ is a symmetric function such that 
\begin{equation}
\label{rncond}r_N(x,y)N^{1/d}\to \infty
\end{equation} uniformly on $A\times A$ as $N\to \infty$, and  $\Phi:(0,\infty)\to [0,\infty)$ is bounded and satisfies $\lim_{t\to 0^+}\Phi(t)=1$.  
For $N\ge 1$, let
\begin {equation}\label {rxy}
v_N(x,y):=\Phi\(\frac {\left|x-y\right|}{r_N(x,y)}\)w(x,y),\ \ \ x,y\in A,\ \ x\neq y.
\end {equation}
Then the conclusions of Theorem~\ref{C2} hold. 

If, in addition,  each $v_N$  is lower semi-continuous, $w$ is bounded, and $A$ is contained in a  $d$-regular set $\tilde A\subset\RR^p$, then every sequence   of  $N$-point $(v_N,s)$-energy minimizing   configurations on $A$  is quasi-uniform on $A$.  
\end{corollary}

Indeed with $v_N$, $r_N$ and $\Phi$ as in the above corollary, it is easy to verify that  $v_N$  satisfies \eqref{1"} and \eqref{1'}.

\bigskip

Finally, we further elucidate the behavior of a sequence of weights $\{v_N\}$ satisfying conditions of form \eqref {1'} in Theorem~\ref{C2}. Given a non-negative function $v$ on $A\times A\setminus D(A)$ and a point $x_0\in A$, let
$$
L(v,x_0):=\limsup _{(x,y)\to (x_0,x_0)\atop x\neq y}{v(x,y)}\ \ \ {\rm and }\ \ \ l(v,x_0):=\liminf _{(x,y)\to (x_0,x_0)\atop x\neq y}{v(x,y)}.
$$
\begin {proposition}\label {P3}
Let $w$ be a CPD-weight defined on $A\times A$ and $\{v_N\}$ be a sequence of non-negative functions on $A\times A\setminus D(A)$. If for some positive sequence $\{\alpha_N\}$ that tends to zero, one has
\begin {equation}\label {alpha_N}
\lim_{N\to\infty}{I^w(v_N,\alpha_N)}=\lim_{N\to\infty}{S^w(v_N,\alpha_N)}=1
\end {equation}
(in particular, if condition (\ref {1'}) holds), then
\begin {equation}\label {UCv_n}
\lim\limits_{N\to \infty}{\frac {L(v_N,x_0)}{L(w,x_0)}}=1\ \ \ and \ \ \ \lim\limits_{N\to \infty}{\frac {l(v_N,x_0)}{l(w,x_0)}}=1
\end {equation}
uniformly over $x_0\in A$.

If, in addition, $w$ is continuous on $D(A)$, then condition \eqref {alpha_N} holds for some positive sequence $\{\alpha_N\}$ with zero limit if and only if
\begin {equation}\label {x_0}
\lim_{N\to\infty}{L(v_N,x_0)}=\lim_{N\to\infty}{l(v_N,x_0)}=w(x_0,x_0)
\end {equation}
with both sequences converging uniformly for $x_0\in A$.
\end {proposition}

The proof of Proposition~\ref{P3} is given in  Appendix~\ref{P3proof}. 

\medskip

If $L(w,x_0)=\infty$ in \eqref {UCv_n}, we agree that $L(v_N,x_0)/L(w,x_0)=0$ when $L(v_N,x_0)<\infty$ and $L(v_N,x_0)/L(w,x_0)=1$  when $L(v_N,x_0)=\infty$. A similar agreement is also in place for the lower limits $l(\cdot,x_0)$.

We also remark that if the limit of $w$   at some point $(x_0,x_0)\in D(A)$ does not exist, one can construct a sequence of weights $\{v_N\}$ such that \eqref {alpha_N} fails for any positive sequence $\{\alpha_N\}$ converging to zero. This can be done even if $w$ is assumed to be bounded on $A\times A$.

\section{Complexity estimates and numerical experiments} \label{section:Comp}

Throughout this section we assume that $\Phi$ is a `cutoff' function as in Theorem \ref {Th1}  such that $$\Phi(t)=0 \text{ for $t>1$.}$$  
For such $\Phi$, we consider the complexity of evaluating 
\begin {equation}\label {u^0}
f(x_1,\ldots, x_N):=E^{v_N}_s(\omega_N)=\sum_{(i,j):i\neq j}{\Phi\(\frac {\left|x_i-x_j\right|}{r_N(x_i,x_j)}\)\frac {w(x_i,x_j)}{\left|x_i-x_j\right|^s}}, 
\end {equation}
where $\omega_N=\{x_1,\ldots, x_N\}$. Assuming  $\Phi$, $r_N$, and $w$ are sufficiently smooth, and $A$ is a compact set in $\RR^d$ of positive Lebesgue measure with boundary of measure zero, we also shall consider the complexity of evaluating
 the  gradient of $f$; i.e., the vector in $\RR^{Nd}$ with    $(d  (i-1)+\ell)^{\rm th}$ component given by
  \begin {equation}\label {u^1}
\partial_{x_{i,\ell}} f(x_1,\ldots, x_N)=  2\sum_{ j: j\neq i }\partial_{x_{i,\ell}}\left({\Phi\(\frac {\left|x_i-x_j\right|}{r_N(x_i,x_j)}\)\frac {w(x_i,x_j)}{\left|x_i-x_j\right|^s}}\right), 
\end {equation} where $x_{i,\ell}$ denotes the $\ell^{th}$ component of $x_i$
for $i=1,\ldots, N$ and $\ell=1,\ldots, d$, as well as the complexity of evaluating the Hessian of $f$; i.e.,    the $Nd\times Nd$ matrix with $(d  (i-1)+\ell,d (j-1)+k)$ component  given by
\begin {equation}\label {u^2}
\partial_{x_{i,\ell}}\partial_{x_{j,k}} f(x_1,\ldots, x_N)=   2\partial_{x_{i,\ell}}\partial_{x_{j,k}}\left({\Phi\(\frac {\left|x_i-x_j\right|}{r_N(x_i,x_j)}\)\frac {w(x_i,x_j)}{\left|x_i-x_j\right|^s}}\right)  
\end {equation}
for $1\le i\neq j\le N$ and $\ell,k=1,\ldots, d$,  and 
\begin {equation}\label {u^3}
\partial_{x_{i,\ell}}\partial_{x_{i,k}} f(x_1,\ldots, x_N)=   2\sum_{ j: j\neq i }\partial_{x_{i,\ell}}\partial_{x_{i,k}}\left({\Phi\(\frac {\left|x_i-x_j\right|}{r_N(x_i,x_j)}\)\frac {w(x_i,x_j)}{\left|x_i-x_j\right|^s}}\right)  
\end {equation}
for $i=1,\ldots, N$ and 
$\ell,k=1,\ldots, d$.

The number of non-zero terms in \eqref{u^0} of the form 
\begin {equation}\label {w*}
\Phi\(\frac {\left|x_i-x_j\right|}{r_N(x_i,x_j)}\)\frac {w(x_i,x_j)}{\left|x_i-x_j\right|^s}
\end {equation} 
does not exceed  the cardinality of $\{(x,y)\in \omega_N\times \omega_N : 0<\left|x-y\right|\leq r_N(x,y)\}$, and so, if $r_N(x,y)\le \delta_N$ for all $x,y\in A$, then    the quantity
\begin {equation}\label {Zdef}
Z(\omega_N,\delta_N):=\# \{(x,y)\in \omega_N\times \omega_N : 0<\left|x-y\right|\leq \delta_N\},
\end {equation} 
 times the maximal complexity of evaluating a single term 
provides an upper bound for the complexity of computing $E^{v_N}_s(\omega_N)$.   Similarly, the number of nonzero terms of the form 
$$
\partial_{x_{i,\ell}}\left({\Phi\(\frac {\left|x_i-x_j\right|}{r_N(x_i,x_j)}\)\frac {w(x_i,x_j)}{\left|x_i-x_j\right|^s}}\right)
$$
required to compute the gradient of $f$ is  bounded above by  $dZ(\omega_N,\delta_N)$, while the number of nonzero elements of the Hessian (each of the form in \eqref{u^2} or \eqref {u^3}) is bounded above by  $2d^2Z(\omega_N,\delta_N)$.  Hence, the computational complexity of one step in a gradient descent optimization scheme (or to evaluate $f$ and its gradient and Hessian, as required in one step of a second-order optimization scheme) is bounded by a constant (determined by the maximal complexity of the individual terms and the dimension $d$) times  $Z(\omega_N,\delta_N)$.  Finally, we mention that  determining the set $\{(x,y)\in \omega_N\times \omega_N : 0<\left|x-y\right|\leq \delta_N\}$ is  known as the {\em fixed-radius near neighbor problem} which can be solved using so-called {\em bucketing algorithms} with (expected) complexity of order $O(N+Z(\omega_N,\delta_N))$   (cf. \cite{BenStaWil77}). 
 
 We further provide  bounds on  $Z(\omega_N,\delta_N)$ based on geometrical and/or energy properties of $\omega_N$.  In order to use Theorem~\ref{Th1} and Corollary~\ref{CorC2}, we must have that $\delta_N$ is of the form  $\delta_N=C_N N^{-1/d}$, for some positive sequence $C_N$    with   infinite limit and we shall assume this form in the following.  We
first observe that
\begin{equation}\label{Zbnd}
Z(\omega_N,\delta_N)\le N\max_{x\in \omega_N}{\#\(\omega_N \cap  B[x,\delta_N]\)},
\end{equation}
 where $B[x,r]$ denotes the closed ball in $\RR^p$ with radius $r$ and center  $x$.  Hence, if  $\{\omega_N\}$ is a sequence of $N$-point configurations on $A$ such that
\begin{equation}\label{Zbnd2}
\max_{x\in \omega_N}{\#\(\omega_N \cap   B[x,\delta_N]\)}=O(N\delta_N^d)=O(C_N^d),\ \ \ N\to \infty,
\end{equation}
then $Z(\omega_N,\delta_N)=O(NC_N^d)$, $N\to \infty$.

If $A$ is a compact subset of $\RR^d$ with boundary of positive Lebesgue measure or $A$ is a $d$-regular subset of $\RR^{p}$, we can still estimate the number of non-zero terms in \eqref {u^0} of form \eqref {w*}. We can show that  a well-separated sequence of configurations $\omega_N$ on a compact $d$-regular set  $A$ satisfies \eqref{Zbnd2} and so we  obtain:
\begin {proposition}\label {P2'}
Let $A$ be a compact $d$-regular set in $\RR^p$, $d\leq p$, and $\{\omega_N\}$ be a sequence of $N$-point configurations on $A$ such that 
\begin {equation}\label {well}
\liminf_{N\to \infty}{\delta (\omega_N)N^{1/d}}>0.
\end {equation}
If $\delta_N=C_N N^{-1/d}$, where $\{C_N\}$ is a positive sequence bounded below by some $c>0$, then
$$
Z(\omega_N,\delta_N)=O(NC_N^d), \ \ \ N\to \infty.
$$
\end {proposition}

\begin{figure}[htbp]
\begin{center}

\includegraphics[scale=.7]{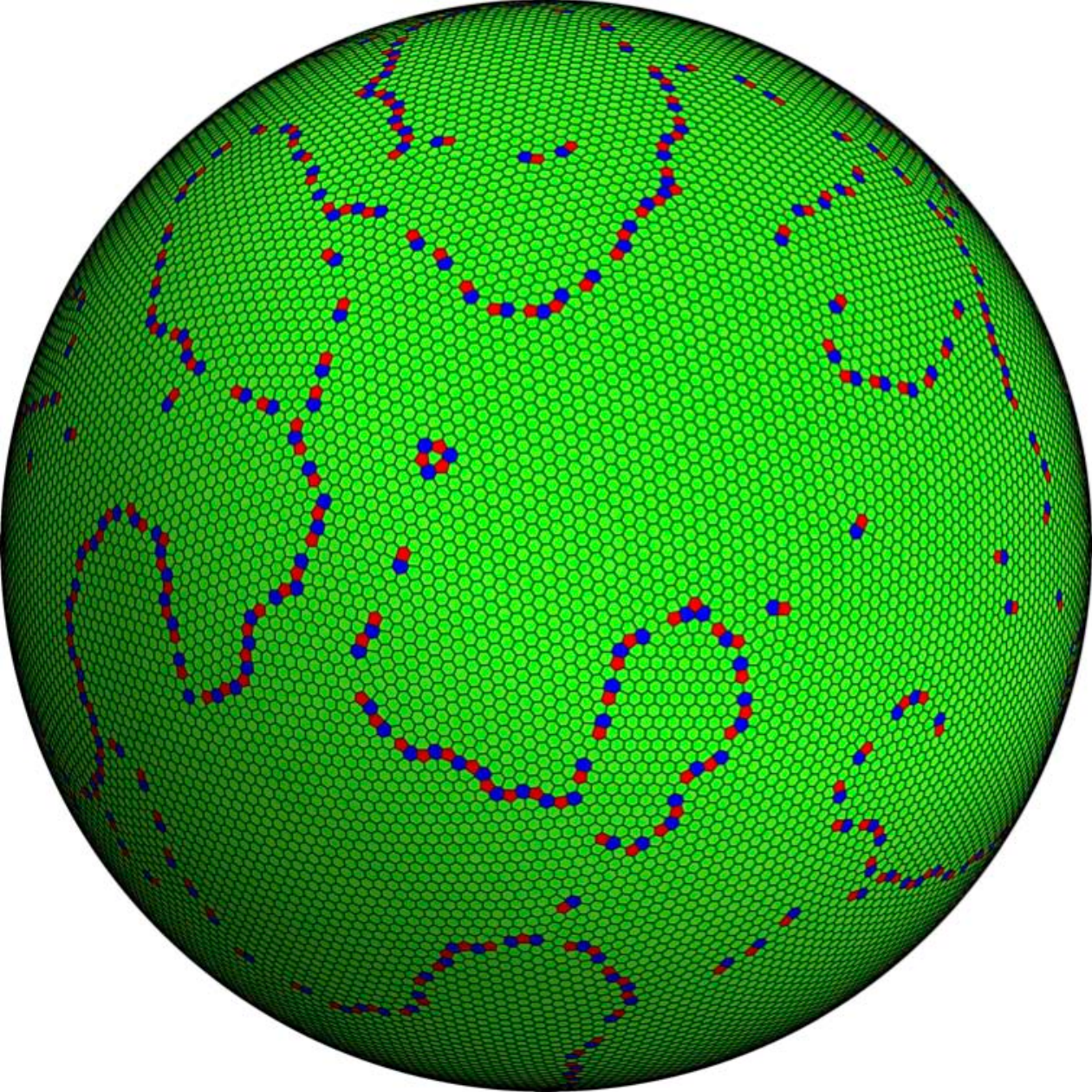}

\caption{{\bf A configuration of 30,000 near optimal $s=3.5$ energy points on the sphere.}}
\label{fig1}
\end{center}
\end{figure}

The following estimate is the most important for our applications to calculating low energy configurations.  
 \begin {proposition}\label {P5}
Let $s>0$, $A$ be a compact set in $\RR^{p}$ and $\omega$ be arbitrary finite configuration on $A$. Then
$$
Z(\omega,\delta)\leq \delta^sE_s(\omega).
$$ 
In particular, if $s>d>0$ and a sequence $\{\omega_N\}$ of $N$-point configurations on $A$ is such that 
$$
E_s(\omega_N)=O(N^{1+s/d}), \ \ \ N\to \infty, \ \ {\rm and}\ \  \delta_N=C_N N^{-1/d}, 
$$
where $\{C_N\}$ is a positive sequence bounded below by some $c>0$, then
$$
Z(\omega_N,\delta_N)=O(NC_N^s),\ \ \ N\to \infty.
$$
\end {proposition}
\noindent
{\bf Remark.}
Note that if $w$ is a bounded CPD weight on $A\times A$ and $s>d$, then $E_s(\omega_N)=O(N^{1+s/d})$ if and only if $E_s^w(\omega_N)=O(N^{1+s/d})$   and so either of these energies can be used in the assumptions of Proposition~\ref{P5}.

To illustrate the utility of our results, we present  two examples of low-energy discretizations.  The first, shown in Figure~\ref{fig1},  shows the Voronoi decomposition of the   unit sphere $\mathbb{S}^2$ for a configuration of 30,000 points on the sphere obtained from a random starting configuration followed by 500 iterations of gradient descent.  We used $s=3.5$, $w(x,y)=1$, $\Phi(t)=(1-t^2)^3\chi_{[0,1]}(t)$, where $\chi_{[0,1]}(t)$ is the characteristic function of the  interval $[0,1]$, and $r_N(x,y)=(\ln N)N^{-1/2}$ 
($\ln N\approx 10$ for $N=30,000$).   We observe (and this is almost always the case for large low energy configurations on the sphere) that  all of Voronoi cells are either pentagons, hexagons, or heptagons, with the large majority being nearly regular hexagons.   This hexagonal dominant local structure lends support to the conjectured value  of $C_{s,2}$ given in the discussion following Theorem~\ref{packing}.

The second example consists of a configuration of $500,000$ low energy points computed in a 3-dimensional spherical shell with inner radius $R_0=.55$ and outer radius $R_1=1$.  We used the same  $s$, $w$, and $\Phi$ as in the previous example.  In this case we chose $r_N(x,y)=(1/4)(\ln N) N^{-1/3}$.  The configuration was obtained by applying 1000 gradient descent iterations to a random starting configuration.  In Figure~\ref{fig3} we show the energy for the configuration at each iteration step and in Figure~\ref{fig2} we show a portion of the configuration near a slice of the shell for the final 1000-th iteration.

\begin{figure}[htbp]
\begin{center}
\includegraphics[scale=1.2]{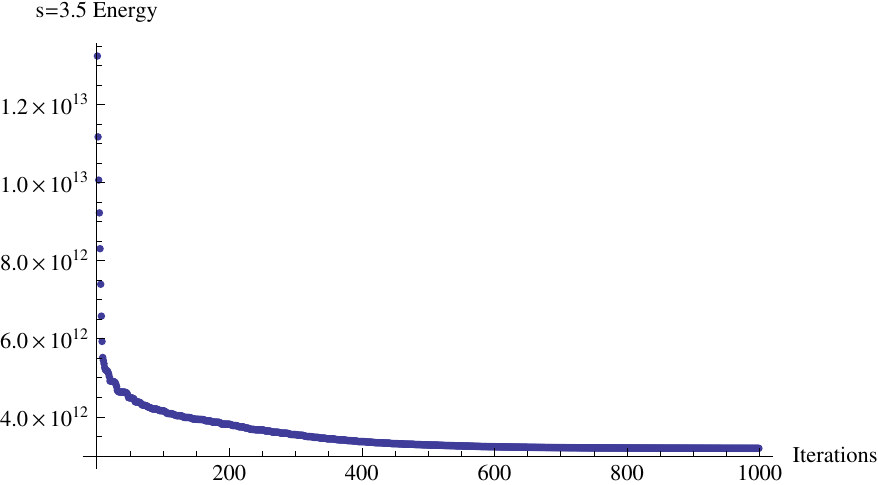}

\caption{{\bf The  energy ($s=3.5$) of a  sequence of 500,000 point configurations in a spherical shell resulting from 1000 gradient descent iterations starting from a random configuration.}}
\label{fig3}
\end{center}
\end{figure}

\begin{figure}[htbp]
\begin{center}
\vspace*{-1.2in}
\hspace*{-1.in}
\includegraphics[scale=.7]{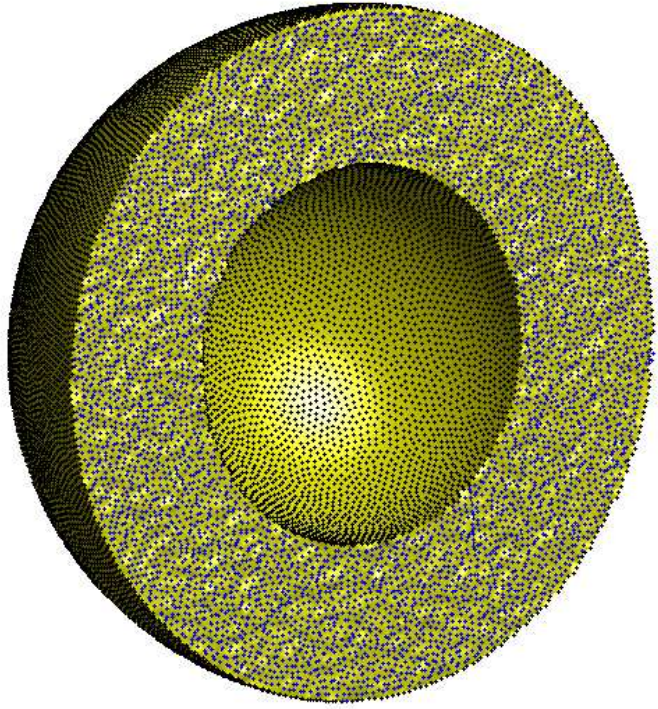}

\vspace*{-2in}
\caption{{\bf A configuration of 500,000 near optimal $s=3.5$ energy points on a spherical shell for geoscience (earth's mantle) applications. }}
\label{fig2}
\end{center}
\end{figure}

\section{Proof of Theorem \ref {rth1}}\label{proof2.1}

Given a sequence of CPD-weight functions $\mathbf{v}=\{v_N\}$ on $A\times A$, let
$$
\UL g^{\mathbf{v}}_{s,d}(A):=\liminf_{N\to\infty}{\frac {\mathcal E^{v_N}_s(A,N)}{N^{1+s/d}}}, \quad
\OL g^{\mathbf{v}}_{s,d}(A):=\limsup_{N\to\infty}{\frac {\mathcal E^{v_N}_s(A,N)}{N^{1+s/d}}},
$$
and, if the limit (possibly infinite) exists, 
$$
g^{\mathbf{v}}_{s,d}(A):=\lim_{N\to\infty}{\frac {\mathcal E^{v_N}_s(A,N)}{N^{1+s/d}}}.
$$
For a constant sequence $v_N=w$, we write $\UL g^{w}_{s,d}(A)$, $\OL g^{w}_{s,d}(A)$, $  g^{w}_{s,d}(A)$ for these respective quantities.  
 In particular, when $v_N\equiv 1$, we omit the superscript.

We shall need the following known results from geometric measure theory.
\begin {theorem}\label {Fed}{\rm (}\cite [Theorem 3.2.39]{FedGMT}{\rm )}
If $W\subset \RR^{p}$ is a closed $d$-rectifiable
set, then $\mathcal M_d(W) =\mathcal H_d(W)$.
\end {theorem} 
A mapping $\varphi: K\to \RR^{p}$, $K\subset \RR^d$, is called {\em bi-Lipschitz} with constant $L>1$ if
$$
L^{-1}\left|x-y\right|\leq \left|\varphi(x)-\varphi(y)\right|\leq L\left|x-y\right|,\ \ \ x,y\in K.
$$
\begin {lemma}\label {Fed1}(see \cite[3.2.18]{FedGMT})
Let $W \subset \RR^{p}$ be an $(\mathcal H_d,d)$-rectifiable set.
Then for every $\epsilon > 0$, there exist (at most countably many) compact sets $K_1, K_2, K_3, \ldots \subset \RR^d$ and
bi-Lipschitz mappings  $\psi_i : K_i \to \RR^{p}$ with constant $1+\epsilon$, $i = 1, 2, 3, \ldots$, such
that  $\psi_1(K_1), \psi_2(K_2),  \psi_3(K_3), \ldots $ are disjoint subsets of $W$ with
$$
\mathcal H_d\(W \setminus\bigcup_i \psi_i(K_i)\)= 0.
$$
\end {lemma}
We start by proving the following auxiliary statement.
\begin {lemma}\label {K}
Let $A\subset \RR^{p}$ be a compact $(\mathcal H_d,d)$-rectifiable set with $\mathcal M_d(A)=\mathcal H_d(A)$. Then every compact subset $K\subset A$ is also $(\mathcal H_d,d)$-rectifiable and $\mathcal M_d(K)=\mathcal H_d(K)$.
\end {lemma}
\begin{proof} Let $K\subset A$ be compact.  It is not difficult to verify that $K$ is also $(\mathcal H_d,d)$-rectifiable. Then, if $\epsilon>0$, it follows from the definition  of $(\mathcal H_d,d)$-rectifiability that there is some $d$-rectifiable compact set $J\subset K$ such that 
$\mathcal H_d(J) \ge \mathcal H_d(K)-\epsilon$.   Using Theorem~\ref{Fed}, we obtain
$$
\UL {\mathcal M}_d(K)\geq \mathcal M_d(J)=\mathcal H_d(J)\ge \mathcal H_d(K)-\epsilon, 
$$
and so $\UL {\mathcal M}_d(K)\geq \mathcal H_d(K)$.

It remains to show that  $\OL {\mathcal M}_d(K)\le \mathcal H_d(K)$. 
Let $\epsilon>0$.
Since, by assumption  $\OL {\mathcal M}_d(A)<\infty$, we must have $\OL {\mathcal M}_d(K)<\infty$. 
Hence, we can find a compact set $L\subset A\setminus K$   such that $\mathcal H_d(L)+\mathcal H_d(K)>\mathcal H_d(A)-\epsilon$. Since ${\rm dist}(K,L)>0$, it is easily seen that   $\OL {\mathcal M}_d(K\cup L)\ge \OL {\mathcal M}_d(K)+\UL {\mathcal M}_d(L)$
and so, using the first part of this proof, we obtain $${\mathcal H}_d(A)={\mathcal M}_d(A)\ge \OL {\mathcal M}_d(K\cup L)\ge \OL {\mathcal M}_d(K)+\UL {\mathcal M}_d(L)\ge \OL {\mathcal M}_d(K)+  {\mathcal H}_d(L).$$ Hence, 
$$\OL {\mathcal M}_d(K)\le {\mathcal H}_d(A)- {\mathcal H}_d(L)\le  {\mathcal H}_d(K)+\epsilon,$$
and it follows that $\OL {\mathcal M}_d(K)\le \mathcal H_d(K)$.
\end{proof}

In view of Lemma \ref {K}, applying Theorem \ref {packing}, we have
$g_{s,d}(K)=C_{s,d}\mathcal H_d(K)^{-s/d}$ for every compact subset $K\subset A$.
Hence, Theorem \ref {rth1} is a consequence of the following known lemma. 
\begin {lemma}\label{wth1}(see \cite [Lemma 6]{BorHarSaf08})
Suppose that $s>d$, $A$ is a compact set in $\RR^{p}$
with $\mathcal H_d (A)<\infty$, and that $w$ is a CPD-weight function on $A\times A$. Furthermore, suppose that for
any compact subset $K\subset A$, the limit $g_{s,d}(K)$ exists and is given by
\begin {equation}\label {w*1}
g_{s,d}(K)=
\frac{C_{s,d}}{\mathcal H_d (K)^{s/d}}.
\end {equation}
Then $g^w_{s,d}(A)$ exists and is given by
\begin {equation}\label {www1}
g^w_{s,d}(A)= C_{s,d}\(\mathcal H^{s,w}_{d}(A)\)^{-s/d}.
\end{equation}
Moreover, if a sequence
$\{\W\omega_N\}_{N=2}^\infty$, where
$\W\omega_N=\{x^N_1,\ldots,x^N_N\}$, is asymptotically
($w,s$)-energy minimizing on the set $A$ and $\mathcal H_d(A)>0$, then
\begin {equation}\label{www2}
\frac 1N\sum_{k=1}^{N}{\delta _{x^N_k}}\stackrel{*}{\longrightarrow} h^{s,w}_{d},\ \ N\to \infty.
\end{equation}
\end{lemma}

\section {A special case of Theorem \ref {Th1}}\label {proof2.2}

We first establish the following special case of Theorem \ref {Th1}.

\begin {proposition}\label {Th1''}
With the assumptions of Theorem \ref {Th1} and the additional hypotheses that {\rm(a)} $\Phi(t)\leq 1$ for $t\in (0,\infty)$ and 
{\rm (b)} $\Phi(t)=0$ for $t>1$, the conclusions of Theorem \ref {Th1} hold; i.e., for $s>d$, we have
\begin {equation}\label {asympt}
g^{\mathbf{v}}_{s,d}(A)=\lim_{N\to \infty}{\frac {\mathcal E^{v_N}_{s}(A,N)}{N^{1+s/d}}}=\frac {C_{s,d}}{\[\mathcal H_d^{s,w}(A)\]^{s/d}},
\end {equation}
where the constant $C_{s,d}$ is as in Theorem~\ref{packing}.
Furthermore, any sequence of  
  asymptotically
$(\mathbf{v},s)$-energy minimizing $N$-point configurations on $A$ is uniformly distributed with respect to the probability measure $h_d^{s,w}$
as $N\to \infty$. 
\end {proposition}

We start by proving the following basic estimate.
\begin {lemma}\label {zeta}
Let $s>d$ and $\omega\subset \RR^d$ be a point configuration such that
$
\delta (\omega)\geq a>0.
$
Then for every $R>a$ and $x\in \omega$,
$$
U_s(\omega;x,R):=\sum_{y\in \omega \atop \left|y-x\right|>R}{\frac {1}{\left|y-x\right|^s}}\leq \frac {d5^{d}}{a^s}\sum_{i=\[R/a\]}^{\infty}{\frac {1}{i^{s-d+1}}}.
$$
\end {lemma}
\begin{proof} For every point $x\in \omega$, let  
$$
T_i(x)=\{y\in \omega : ai\leq\left|y-x\right|< a(i+1)\},\ \ \ i\in \NN.
$$
Then since the collection of open balls of radius $a/2$ centered at points of $\omega$ is pairwise disjoint, we have
\begin{equation}
\begin{split}
\# T_i(x)&\leq \frac {\mathcal L_d\(B(x,(i+3/2)a)\setminus B(x,(i-1/2)a)\)}{\mathcal L_d(B({\bf 0},a/2))}\\
&=\frac {(i+3/2)^da^d-(i-1/2)^da^d}{(a/2)^d}
=(2i+3)^d-(2i-1)^d\\
&\leq 4d (2i+3)^{d-1}\leq d 5^{d}i^{d-1},\ \ \ i\in \NN.
\end{split}
\end{equation}Hence,
$$
U_s(\omega;x,R)\leq \sum_{i=\[R/a\]}^{\infty}{\sum_{y\in T_i(x)}{\frac {1}{\left|y-x\right|^s}}}\leq  \sum_{i=\[R/a\]}^{\infty}\frac {\# T_i(x)}{(ai)^s}\leq  \frac {d5^{d}}{a^s}\sum_{i=\[R/a\]}^{\infty}\frac {1}{i^{s-d+1}},
$$
which concludes the proof.
\end{proof}

\begin{proof}[Proof of Proposition~\ref{Th1''}.]
From \eqref{v} and the additional hypotheses on $\Phi$ we have $v_N(x,y)\leq w(x,y)$, $x,y\in A$, $x\neq y$. Hence, in view of Theorem \ref {rth1}, there holds for $s>d$, 
\begin {equation}\label {upper}
\OL g^{\mathbf{v}}_{s,d}(A)=\limsup_{N\to \infty}{\frac {\mathcal E^{v_N}_{s}(A,N)}{N^{1+s/d}}}\leq \lim_{N\to \infty}{\frac {\mathcal E^{w}_{s}(A,N)}{N^{1+s/d}}}=\frac {C_{s,d}}{\left[\mathcal H_d^{s,w}(A)\right]^{s/d}}, 
\end {equation}
 proving the upper estimate for (\ref {asympt}). \\

{\bf Bounded weight.} We first establish the required  lower bound for (\ref {asympt}) under the assumption that the CPD-weight $w$ is bounded on $A\times A$. 
Let $h$ and $\kappa$ be positive numbers such that $w(x,y)>h$ whenever $x,y\in A$ and $\left|x-y\right|<\kappa$. Such numbers $h$ and $\kappa$ exist in view of condition (b) in the definition of the CPD-weight. 
Define
\begin {equation}\label {OLPhi}
\OL \Phi(t):=\inf_{u\in (0,t]}{\Phi(u)},\ \ \ t>0.
\end {equation}
Let $\{\omega_N\}$ be any sequence of point configurations on $A$ such that $\# \omega_N=N$ and
\begin {equation}\label {*}
\left|\mathcal E^{v_N}_{s}(A,N)-E^{v_N}_{s}(\omega_N)\right|=o(N^{1+s/d}),\ \ \ N\to\infty.
\end {equation}

Let \begin{equation}\label{CN}
C_N:=r_N N^{1/d} \qquad (N=2,3,\ldots).
\end{equation} 
Our goal is to show that the total energy of the pairs of points in $\omega_N$ that are at least $\sqrt{C_N}N^{-1/d}$ away from each other is $o(N^{1+s/d})$, from which the lower bound will follow.   The argument consists of the following five steps: 

\emph{Step 1.} For a sufficiently small positive constant $C$, we remove from $\omega_N$ all those points whose $CN^{-1/d}$-neighborhood contains another point from $\omega_N$ and show that the configuration $\omega_{N,C}$ of the remaining points has sufficiently large cardinality.

\emph{Step 2.} We choose a subset $D\subset A$ that consists of finitely many pairwise disjoint bi-Lipschitz embeddings of compact subsets of $\RR^d$ and whose complement with respect to $A$ has small $\mathcal H_d$-measure and
 show that the set $\eta_{N,C}$ of points from $\omega_{N,C}$ that are sufficiently close to $D$ still has a sufficiently large cardinality.

\emph{Step 3.} We move each point in $\eta_{N,C}$ to a close point in $D$ and show that the resulting configuration $z_{N,C}$ has almost the same separation as $\eta_{N,C}$.

\emph{Step 4.} We prove that the total energy of the pairs of points in $z_{N,C}$ that are sufficiently separated from each other is $o(N^{1+s/d})$. Since the bi-Lipschitz pieces of $D$ are metrically separated, only the pairs of points from the same piece will make a significant contribution. This allows us to switch to estimating energies in $\RR^d$ using Lemma \ref {zeta}.

\emph{Step 5.} Since $\lim_{t\to 0^+}\Phi(t)=1$, by varying the constant $C$ the leading term of the $(w,s)$-energy of $\eta_{N,C}$ can be made as close as we like to the leading term of $(v_N,s)$-energy of $\omega_N$ thus giving us a sharp lower estimate for $\mathcal E^{v_N}_s(A,N)$. 
 
For Step 1, choose a number $C\in (0,1/2)$ such that 
\begin{equation}\label{alphaC}
\alpha_C:=3^{d+1} \beta_{p-d}\beta_{p}^{-1}C^{2d}+C^sh^{-1}g^w_{s,d}(A)<1
\end{equation}
and set
$$
\omega_{N,C}:=\{x\in \omega_N : {\rm dist}(x,\omega_N\setminus \{x\})>CN^{-1/d}\},\ \ \ N\in \NN.
$$
Let $y_x$ be a point in $\omega_N\setminus \{x\}$ closest to a given point $x\in \omega_N$.  Then, for every $N$ sufficiently large,
\begin{equation*}
\begin{split}
{E}^{v_N}_s&(\omega_N)=\sum_{x\in \omega_N}{\sum_{y\in \omega_N\setminus \{x\}}{\frac {v_N(x,y)}{\left|x-y\right|^s}}}\geq \sum_{x\in \omega_N\setminus \omega_{N,C}}{\frac {\Phi\(\left|x-y_x\right|r_N^{-1}\)w(x,y_x)}{\left|x-y_x\right|^s}}\\
&\geq  \sum_{x\in \omega_N\setminus \omega_{N,C}}{\frac {\OL\Phi\(CN^{-1/d}r_N^{-1}\)w(x,y_x)}{\left|x-y_x\right|^s}}\geq h\OL\Phi\(\frac {C}{C_N}\)
\sum_{x\in \omega_N\setminus \omega_{N,C}}\frac{1}{\left|x-y_x\right|^s}\\
&
\geq \# (\omega_N\setminus \omega_{N,C})\cdot  hC^{-s}\OL\Phi\(\frac {C}{C_N}\)N^{s/d}.
\end{split}
\end{equation*}
Consequently,
\begin{equation} \label{pf5.1.1}\begin{split}
g^w_{s,d}(A)&=\lim_{N\to \infty}{\frac {\mathcal E^{w}_{s}(A,N)}{N^{1+s/d}}}\geq 
\limsup_{N\to \infty}{\frac {\mathcal E^{v_N}_{s}(A,N)}{N^{1+s/d}}}\\
&=\limsup_{N\to \infty}{\frac {E^{v_N}_s(\omega_N)}{N^{1+s/d}}}\geq hC^{-s}\limsup_{N\to \infty}{\frac {\# (\omega_N\setminus \omega_{N,C})}{N}}.
\end{split}
\end{equation}
By Theorem \ref {rth1}, since $\mathcal H_d(A)>0$, the quantity $g^w_{s,d}(A)$ is finite. Hence, from \eqref{pf5.1.1}, we have
\begin{equation}\label {omegaNC}
\liminf_{N\to \infty}{\frac {\#  \omega_{N,C}}{N}}\geq 1-C^sh^{-1}g^w_{s,d}(A),
\end {equation}
which completes Step 1. 

To proceed with Step 2, let $\delta=C^{4d}$.  In view of Lemma \ref {Fed1}, there exist compact sets $K_1,K_2,\ldots, K_m\subset \RR^d$ and bi-Lipschitz mappings $\psi_i:K_i\to \RR^{p}$ with bi-Lipschitz constant $\lambda$, $i=1,\ldots,m$, such that the set 
$$
D:=\bigcup_{i=1}^{m}\psi_i(K_i) 
$$
is contained in $A$ and satisfies
$$
\mathcal H_d(D)>\mathcal H_d(A)-\delta.
$$
Moreover, $\psi_i(K_i)$, $i=1,\ldots,m$, are pairwise disjoint.
Since each set $\psi_i(K_i)$ is $d$-rectifiable, the set $D$ is also $d$-rectifiable, and by Theorem \ref {Fed},
$$
\mathcal M_d(D)=\mathcal H_d(D)>\mathcal H_d(A)-\delta=\mathcal M_d(A)-\delta.
$$
Let 
$
h_N:= {C^2}/(3N^{1/d}),\ \ N\in \NN,
$
and recall that $A(\epsilon)$ denotes the $\epsilon$-neighborhood of a set $A$ in $\RR^{p}$.
Then for every $N$ sufficiently large, 
$$
\mathcal L_{p}\(A(h_N)\setminus D(h_N)\)=\mathcal L_{p}\(A(h_N)\)-\mathcal L_{p}\( D(h_N)\)\leq 
$$
$$
\leq\beta_{p-d}\(\mathcal M_d(A)+\delta\)h_N^{p-d}-\beta_{p-d}\(\mathcal M_d(D)-\delta\)h_N^{p-d}\leq 3\delta \beta_{p-d}h_N^{p-d}.
$$
Let $\W \eta_{N,C}:=\omega_{N,C}\setminus D(3h_N)$ and
$$
F_N=\bigcup_{x\in \W \eta_{N,C}}{B(x,h_N)},\ \ \ N\in \NN.
$$
Then $F_N\subset A(h_N)\setminus D(h_N)$. Since for every $x,y\in \omega_{N,C}$, $x\neq y$, we have
$$
\left|x-y\right|\geq CN^{-1/d}\geq C^2N^{-1/d}=3h_N,
$$
the collection $\{B(x,h_N) : x\in \W \eta_{N,C}\}$, is pairwise disjoint. Thus
$$
\# \W \eta_{N,C}=(\beta_{p}h_N^{p})^{-1}\mathcal L_{p}(F_N)\leq (\beta_{p}h_N^{p})^{-1}\mathcal L_{p}\(A(h_N)\setminus D(h_N)\)\leq
$$
\begin {equation}\label{uNC}
\leq 3\delta \beta_{p-d}h_N^{p-d}(\beta_{p}h_N^{p})^{-1}=3^{d+1}\beta_{p-d}\beta_{p}^{-1}C^{2d}N.
\end {equation}
Setting $\eta_{N,C}:=\omega_{N,C}\cap D(3h_N)$,  it follows from \eqref{uNC}, \eqref{omegaNC}, and \eqref{alphaC}, that 
 \begin{equation}\label{alphaC2}
 \liminf_{N\to \infty}\frac{\#\eta_{N,C}}{N}\ge 1-C^sh^{-1}g^w_{s,d}(A)-3^{d+1}\beta_{p-d}\beta_{p}^{-1}C^{2d}=1-\alpha_{C},
 \end{equation} 
 which completes Step 2.

 For the next step, let $z:\eta_{N,C}\to D$ be a mapping, where $z(x)$, $x\in \eta_{N,C}$, is a point in $D$ such that $\left|z(x)-x\right|<3h_N$. Then
$$
\left|z(x)-x\right|<C^2N^{-1/d}\leq C\delta (\omega_{N,C})\leq C\delta (\eta_{N,C}),\ \ \ x\in \eta_{N,C},
$$
and for every pair of distinct points $x,y\in \eta_{N,C}$, we have
$$
\left|z(x)-z(y)\right|\geq \left|x-y\right|-\left|z(x)-x\right|-\left|z(y)-y\right|\geq
$$
$$
\geq\left|x-y\right|-2C\delta (\eta_{N,C})\geq (1-2C)\left|x-y\right|>0,
$$
which implies that $z$ is an injective mapping. Similarly,
$$
\left|z(x)-z(y)\right|\leq (1+2C)\left|x-y\right|, \ \ \ x,y\in \eta_{N,C},\ \ x\neq y,
$$
completing Step 3.

We now consider Step 4.  Let $\xi_N:=(1-2C)\sqrt{C_N}N^{-1/d}$ and  $z_{N,C}:=z(\eta_{N,C})=\{z(x): x\in \eta_{N,C}\}$.
Then since $w$ was assumed to be bounded, we have
\begin{equation*}
\begin{split}
\Pi^w_s(\eta_{N,C})&:=\sum_{ {x,y\in \eta_{N,C}}\atop {|x-y|> \sqrt{C_N}N^{-1/d}} }{\frac {w(x,y)}{\left|y-x\right|^s}}\leq \|w\|_{\infty}\sum_{ {x,y\in \eta_{N,C}}\atop {|x-y|> \sqrt{C_N}N^{-1/d}} }{\frac {1}{\left|y-x\right|^s}}\\
 &
\leq (1+2C)^s\|w\|_\infty \sum_{ {x,y\in \eta_{N,C}}\atop {|x-y|> \sqrt{C_N}N^{-1/d}} }{\frac {1}{\left|z(x)-z(y)\right|^s}}\\ &\leq (1+2C)^s\|w\|_\infty \sum_{x,y\in z_{N,C}\atop \left|x-y\right|>\xi_N}{\frac {1}{\left|y-x\right|^s}}.
\end{split}
\end{equation*}

Let $G_i=\psi_i^{-1}(z_{N,C})\cap K_i$, $i=1,\ldots,m$. Then 
$$
\delta(G_i)\geq \frac {1-2C}{\lambda}\delta(\eta_{N,C})\geq \theta_N:=\frac {C(1-2C)}{\lambda N^{1/d}}
$$ 
and since $\psi_i(K_i)$ are pairwise disjoint,
$\sum_{i=1}^{m}{(\# G_i)}=\# \eta_{N,C}$. Since
$$
\tau:=\min_{1\leq i\neq j\leq m}{\rm dist}(\psi_i(K_i),\psi_j(K_j))>0,
$$ 
and $G_i\subset \RR^d$, $i=1,\ldots,m$, taking into account Lemma \ref {zeta}, we have, for $N$ sufficiently large, with $\sigma^s_w:=
(1+2C)^s\|w\|_\infty$
\begin{equation*}
\begin{split}
\Pi^w_s&(\eta_{N,C})\leq \sigma^s_w\(\sum_{i=1}^{m}{\sum_{x,y\in z_{N,C}\cap \psi_i(K_i)\atop \left|x-y\right|>\xi_N}{\frac {1}{\left|y-x\right|^s}}}+\sum_{x,y\in z_{N,C}\atop \left|x-y\right|\geq\tau}{\frac {1}{\left|y-x\right|^s}}\)\\
&
\leq \sigma^s_w\(\lambda^s \sum_{i=1}^{m}{\sum_{x,y\in z_{N,C}\cap \psi_i(K_i)\atop \left|y-x\right|>\xi_N}{\frac {1}{\left|\psi_i^{-1}(x)-\psi_i^{-1}(y)\right|^s}}}+\tau^{-s}N^2\)\\
&\leq \sigma^s_w\(\lambda^s \sum_{i=1}^{m}{\sum_{x\in G_i}\sum_{y\in G_i\atop \left|y-x\right|>\xi_N/\lambda}{\frac {1}{\left|x-y\right|^s}}}+\tau^{-s}N^2\)\
\\&\leq \sigma^s_w\(\lambda^s \sum_{i=1}^{m}{\sum_{x\in G_i}\frac {d5^{d}}{\theta_N^s}\sum_{j=\[\frac {\xi_N}{\lambda\theta_N}\]}^{\infty}{\frac {1}{j^{s-d+1}}}}+\tau^{-s}N^2\)\\
&
=\sigma^s_w\(\lambda^{2s} \sum_{i=1}^{m}{\frac {d5^{d}(\# G_i)N^{s/d}}{C^s(1-2C)^s}\sum_{j=\[\sqrt{C_N}/C\]}^{\infty}{\frac {1}{j^{s-d+1}}}}+\tau^{-s}N^2\)
\\
&
=\sigma^s_w\(\lambda^{2s}\frac {d5^{d}(\# \eta_{N,C})o(N^{s/d})}{C^s(1-2C)^s}+o(N^{1+s/d})\)\\
&=o(N^{1+s/d}), \qquad N\to\infty,
\end{split}
\end{equation*}
which completes Step 4. 

For the last step, we use the above estimates to obtain
\begin{equation} \label {**}
\begin{split}
E^{v_N}_s(\omega_N)&\geq E^{v_N}_s(\eta_{N,C})\geq \sum_{x,y\in \eta_{N,C},\ x\neq y\atop \left|x-y\right|\leq \sqrt{C_N}/N^{1/d}}{\frac {\Phi\(\frac {\left|x-y\right|}{r_N}\)w(x,y)}{\left|x-y\right|^s}}
\\
&\geq \OL \Phi\(\frac {1}{\sqrt{C_N}}\)\sum_{x,y\in \eta_{N,C},\ x\neq y\atop \left|x-y\right|\leq \sqrt{C_N}/N^{1/d}}\frac {w(x,y)}{\left|x-y\right|^s}
\\
&=\OL \Phi\(\frac {1}{\sqrt{C_N}}\)\(E^w_s(\eta_{N,C})-\Pi^w_s(\eta_{N,C})\)
\\
&\geq \OL \Phi\(\frac {1}{\sqrt{C_N}}\)\(\mathcal E^w_s(A,\# \eta_{N,C})+o(N^{1+s/d})\),\ \ \ N\to \infty.
\end{split}
\end{equation}

Then,   taking into account Theorem \ref {rth1},  relations \eqref{*} and \eqref {alphaC2} and the fact that  $\lim_{t\to 0^+}{\OL \Phi(t)}=1$, we have
\begin{equation*}  
\begin{split}
\UL g^{\mathbf{v}}_{s,d}(A)&=\liminf_{N\to\infty}\ \frac {\mathcal E^{v_N}_s(A,N)}{N^{1+s/d}}=\liminf_{N\to\infty}\frac {E^{v_N}_s(\omega_N)}{N^{1+s/d}}\\
&\geq 
\liminf_{N\to\infty}\frac {\OL \Phi\(\frac {1}{\sqrt{C_N}}\)\mathcal E^w_s(A,\# \eta_{N,C})}{N^{1+s/d}}
\\
&\geq\lim_{N\to\infty}\frac {\mathcal E^w_s(A,\# \eta_{N,C})}{(\# \eta_{N,C})^{1+s/d}}\cdot \liminf_{N\to\infty}\(\frac {\# \eta_{N,C}}{N}\)^{1+s/d}
\\
&\geq\frac {C_{s,d}}{\mathcal H^{s,w}_{d}(A)^{s/d}}\(1-\alpha_C\)^{1+s/d}.
\end{split}
\end{equation*}
Letting $C\to 0$ gives 
$
\UL g^{\mathbf{v}}_{s,d}(A)\geq   {C_{s,d}}/{\mathcal H^{s,w}_{d}(A)^{s/d}},
$
 completing Step 5.  
Taking into account (\ref {upper}), we obtain relation (\ref {asympt}) for the case of bounded CPD-weight $w$.\\

{\bf Unbounded weight.} We now prove \eqref{asympt}  for  an   arbitrary (not necessarily bounded) CPD-weight $w$  on $A\times A$.  Let 
$$
w^M(x,y):=\min\{w(x,y),M\}, \ \ x,y\in A, \ \ M>0. 
$$
It is not difficult to see that $w^M$ is also a CPD-weight function on $A\times A$.   Let $\mathbf{u}=\{u_N\}$ denote the  sequence of `truncated' weights
$$
u_N(x,y):=\Phi\(\frac {\left|x-y\right|}{r_N}\)w^M(x,y),\ \ \ x,y\in A,\ \ x\neq y,\ \ \ N\in \NN.
$$
As shown above (\ref {asympt}) holds for bounded CPD-weights, and hence, for every $M>0$, we have
$$
\UL g^{\mathbf{v}}_{s,d}(A)\geq g^{\mathbf{u}}_{s,d}(A)=C_{s,d}\(\int_{A}{(w^M(x,x))^{-d/s}{\rm d}\mathcal H_d(x)}\)^{-s/d}.
$$
Letting $M\to \infty$, we obtain from the Monotone Convergence Theorem that
$$
\UL g^{\mathbf{v}}_{s,d}(A)\geq g^w_{s,d}(A)=C_{s,d}\(\int_{A}{(w(x,x))^{-d/s}{\rm d}\mathcal H_d(x)}\)^{-s/d}=\frac {C_{s,d}}{\left[\mathcal H_d^{s,w}(A)\right]^{s/d}}.
$$
Together with \eqref{upper}, we get (\ref {asympt}) for the case of an unbounded weight.\\

\noindent {\bf Remark.} It is easy to see that (\ref {asympt}) holds if $r_N$ is only defined for a subsequence $\mathcal N\subset \NN$, a fact that we shall use in the next part of the proof.
In this case, we shall also use $g^{\mathbf{v}}_{s,d}(A)$ to denote the limit along this subsequence.\\

We next prove the limit distribution assertion in Proposition \ref {Th1''}.   
Let $\{\omega_N\}$ be an asymptotically $\(\mathbf{v},s\)$-energy minimizing sequence of $N$-point configurations on $A$.   It might appear to the reader that a simple argument would show that $\{\omega_N\}$ is also asymptotically $(w,s)$-energy minimizing so that the limiting distribution statement in Theorem~\ref{rth1} may be applied.  However, the authors have not as yet found such an argument.  Instead, we adapt  methods in \cite{HarSaf05} and \cite{BorHarSaf08} to the varying weight case.  

Let $B$ be an arbitrary \emph{almost clopen} subset of $A$, that is,     the boundary $\partial_A B$ of  $B$ relative to $A$ has $\mathcal H^{s,w}_d$-measure zero.  Since $B$ is an arbitrary almost clopen subset of $A$, the condition that $\omega_N$ is uniformly distributed with respect to $h^{s,w}_d$ is equivalent to   \begin {equation}\label {distlim}
\lim_{N\to \infty}{\frac {\# (\omega_N\cap B)}{N}}= h^{s,w}_d(B).
\end {equation}

     By Lemma \ref {K}, both the closure $\OL B$ of $B$ and the closure $\OL {D}$ of $D:=A\setminus B$ are compact $(\mathcal H_d,d)$-rectifiable sets, for which the $d$-dimensional Minkowski content exists and coincides with the $\mathcal H_d$-measure.

  We  consider the case $0<\mathcal H_d^{s,w}(B)<\mathcal H^{s,w}_d(A)$ and leave the cases $\mathcal H_d^{s,w}(B)=0$ or $\mathcal H_d^{s,w}(B)=\mathcal H^{s,w}_d(A)$ to the reader. Then both $\mathcal H_d^{s,w}(\OL B)$ and $\mathcal H_d^{s,w}(\OL {D})$ are positive. 
Let   $\mathcal N\subset \NN$ be an infinite subset such that the limit
\begin {equation}\label {a}
\alpha:=\lim_{\mathcal N\ni N\to\infty}{\frac {\# (\omega_N\cap B)}{N}}
\end {equation}
exists. Denote $N_B=\# (\omega_N\cap B)$ and $N_D=\# (\omega_N\setminus B)$. Then for every $N\in \mathcal N$, we have
$$
E^{v_N}_s(\omega_N)\geq E^{v_N}_s(\omega_N\cap B)+E^{v_N}_s(\omega_N\setminus B)\geq \mathcal E^{v_N}_s(\OL B,N_B)+\mathcal E^{v_N}_s(\OL D,N_D).
$$
If $\alpha\in (0,1]$, 
denote by $\mathcal N_1\subset \mathcal N$ an infinite subset such that the sequence $\{\# (\omega_N\cap B)\}_{N\in \mathcal N_1}$ is strictly increasing. Let also $\mathcal M_1=\{\# (\omega_N\cap B) : N\in \mathcal N_1\}$. 

If $\alpha\in [0,1)$, we further
let $\mathcal N_2 \subset \mathcal N_1 $ (if $\alpha=0$, let $\mathcal N_2 \subset \mathcal N$) be an infinite subset such that 
 the sequence $\{\# (\omega_N\setminus B)\}_{N\in \mathcal N_2}$ is strictly increasing. Let also $\mathcal M_2=\{\# (\omega_N\setminus B) : N\in \mathcal N_2\}$.

Denote by $n(M)$, $M\in \mathcal M_1$, the unique integer from $\mathcal N_1$ such that $\# (\omega_{n(M)}\cap B)=M$ and let $k(M)$, $M\in \mathcal M_2$ be the unique integer from $\mathcal N_2$ such that $\# (\omega_{k(M)}\setminus B)=M$. 
Note that if $\alpha\in (0,1]$, in view of assumption (\ref {a}),
$$
r_{n(M)}=\frac {C_{n(M)}}{n(M)^{1/d}}
=\frac {\alpha^{1/d} C_{n(M)}}{M^{1/d}}(1+o(1)),\ \ \ \mathcal M_1\ni M\to \infty,
$$
where $C_{n(M)}\to \infty$, $\mathcal M_1\ni M\to \infty$. Analogously, if $\alpha\in [0,1)$, we have
$$
r_{k(M)}=\frac {C_{k(M)}}{k(M)^{1/d}}=\frac {(1-\alpha )^{1/d} C_{k(M)}}{M^{1/d}}(1+o(1)),\ \ \ \mathcal M_2\ni M\to \infty,
$$
where $C_{k(M)}\to \infty$, $\mathcal M_2\ni M\to \infty$.

In view of relation (\ref {asympt}), for any positive sequence $\{\kappa_N\}$ satisfying $\lim_{N\to \infty}\kappa_NN^{1/d}=\infty$, we have
\begin{equation}\label{gu}
g^{\mathbf{u}}_{s,d}(V)=g^w_{s,d}(V)=\frac{C_{s,d}}{\left[\mathcal H^{s,w}_d(V)\right]^{s/d}}, \quad \text{($V=\OL B$ or $V=\OL {D}$)},
\end{equation}
where $\mathbf{u}=\{u_N\}$ is given by 
$
u_N(x,y)=\Phi\(  {\left|x-y\right|}/{\kappa_N}\)w(x,y).
$

Suppose $\alpha\in (0,1)$. 
Applying relation (\ref {asympt}) to the set $A$,   using (\ref {gu}) with $\kappa_M= r_{n(M)}$ for $M\in\mathcal{M}_1$ (respectively, $\kappa_M= r_{k(M)}$ for $M\in\mathcal{M}_2$) and  $V=\OL B$ (respectively, $\OL D$),   
and taking into account that $\mathcal N_2\subset \mathcal N_1$, we  obtain
\begin{equation*} 
\begin{split}
 \frac {C_{s,d}}{\left[\mathcal H^{s,w}_d(A)\right]^{s/d}}  &=\lim_{N\to \infty}{\frac {E^{v_N}_s(\omega_N)}{N^{1+s/d}}}
= \lim_{\mathcal N_2\ni N\to\infty}{\frac {E^{v_N}_s(\omega_N)}{N^{1+s/d}}}
\\
&\geq \liminf_{\mathcal N_2 \ni N\to\infty}{\frac {\mathcal E^{v_N}_s(\OL B,N_B)}{N_B^{1+s/d}}\cdot \(\frac {N_B}{N}\)^{1+s/d}}
\\
&\hspace{1in} +\liminf_{\mathcal N_2 \ni N\to\infty}{\frac {\mathcal E^{v_N}_s(\OL D,N_D)}{N_D^{1+s/d}}\cdot \(\frac {N_D}{N}\)^{1+s/d}}
\\
&\geq\alpha^{1+s/d}\liminf_{\mathcal N_1\ni N\to\infty}{\frac {\mathcal E^{v_N}_s(\OL B,N_B)}{N_B^{1+s/d}}}
\\
&\hspace{1in}+(1-\alpha)^{1+s/d}\liminf_{\mathcal N_2\ni N\to\infty}{\frac {\mathcal E^{v_N}_s(\OL D,N_D)}{N_D^{1+s/d}}}
\\
&=\alpha^{1+s/d}\lim_{\mathcal M_1\ni M\to\infty}{\frac {\mathcal E^{v_{n(M)}}_s(\OL B,M)}{M^{1+s/d}}}
\\
&\hspace{1in}
+(1-\alpha)^{1+s/d}\lim_{\mathcal M_2\ni M\to\infty}{\frac {\mathcal E^{v_{k(M)}}_s(\OL D,M)}{M^{1+s/d}}}
\\
&
=C_{s,d} \(\frac{\alpha^{1+s/d}}{\mathcal H^{s,w}_d(B)^{s/d}}+\frac{(1-\alpha)^{1+s/d}}{\mathcal H^{s,w}_d(D)^{s/d}}\)=:F(\alpha).
\end{split}
\end {equation*}
We remark that if $\alpha=0$ or $\alpha=1$, then appropriate terms may be dropped and the final inequality   still holds.
Furthermore, it is not difficult to see that the minimum value of $F$ on [0,1] is given by ${C_{s,d}}{\left[\mathcal H^{s,w}_d(A)\right]^{-s/d}}$ and occurs only at the   point
$$
\W \alpha:=\frac {\mathcal H^{s,w}_d(B)}{\mathcal H^{s,w}_d(B) +\mathcal H^{s,w}_d(D)}=h^{s,w}_d(B).
$$
Hence, the above inequality shows   $\alpha=\W\alpha$. Since $\mathcal N\subset \NN$ is arbitrary, we obtain (\ref{distlim}), which completes the proof of Proposition~\ref{Th1''}.

%

\end{proof}

 We shall use the following corollary in the proof of Theorem~\ref{C2}.
\begin {corollary}\label {C1}
Let $s>d$, $A\subset \RR^{p}$, $p\geq d$, be a compact $(\mathcal H_d,d)$-rectifiable set with $\mathcal H_d(A)=\mathcal M_d(A)>0$, $w$ be a CPD-weight
function on $A\times A$, and $\{r_N\}$ be a positive sequence satisfying (\ref {r_N}).
For any asymptotically $(w,s)$-energy minimizing sequence $\{\omega_N\}$ of point configurations on $A$ such that $\# \omega_N=N$, there holds
$$
P^w_s(\omega_N,r_N):=\sum_{x,y\in \omega_N \atop \left|x-y\right|>r_N}{\frac {w(x,y)}{\left|x-y\right|^s}}=o(N^{1+s/d}),\ \ \ N\to\infty.
$$
\end {corollary}

\begin{proof}
  Let $\Phi_0:=\chi_{[0,1]}$ be the characteristic function of $[0,1]$ and let $\mathbf{v}=\{v_N\}$ be as in Proposition~\ref {Th1''} with $\Phi$ replaced by $\Phi_0$.  Since $\{\omega_N\}$ is an asymptotically $(w,s)$-energy minimizing,   it is also asymptotically $(\mathbf{v},s)$-energy minimizing.  Observing that 
  $$E_s^w(\omega_N)=E_s^{v_N}(\omega_N)+P^w_s(\omega_N,r_N),\quad (N\in \NN)$$
and  applying Proposition~\ref {Th1''} completes the proof.  

\end{proof}

\section {Proof of Theorems~\ref{Th1} and \ref{C2}.} \label{proof2.2b}

We shall first prove Theorem~\ref{C2} from which we will deduce Theorem~\ref{Th1}.

\begin{proof}[Proof of Theorem \ref {C2}.] We first show that condition (\ref {1'}) implies that there is a positive sequence $\{r_N\}$ satisfying (\ref {r_N}) such that
\begin {equation}\label {I}
\lim_{N\to\infty}{I^w(v_N,r_N)}=\lim_{N\to\infty}{S^w(v_N,r_N)}=1.
\end {equation}
Indeed, for every $K\in \NN$, one can choose a number $N_K\in \NN$ such that
$$
\left|I^w(v_N,KN^{-1/d})-1\right|<\frac {1}{K},\ \ \ \text{ for every }\ \ N>N_K,
$$
and that $N_1<N_2<N_3<\ldots$. Furthermore, we can increase each $N_k$ so that $N_k>k^{d+1}$ and $\{N_k\}$ is still an increasing sequence.  Define a sequence $\{C_N\}$ in the following way. Let $C_1,\ldots,C_{N_1}$ be arbitrary positive numbers and let $C_N:=1$ for $N_1<N\leq N_2$, $C_N:=2$ for $N_2<N\leq N_3$, ..., $C_N:=m$ for $N_m<N\leq N_{m+1}$, ... . Then since $N>N_{C_N}$ for every $N>N_1$, we have 
$$
\left|I^w(v_N,C_NN^{-1/d})-1\right|<\frac {1}{C_N},\ \ \ \text{ for every }\ \ N>N_1.
$$
Since $C_N\to\infty$, $N\to \infty$, we have
$$
\lim_{N\to\infty}{I^w(v_N,\tau_N)}=1,
$$
where $\tau_N:=C_NN^{-1/d}$. Since $\tau_N=C_NN^{-1/d}<C_NN_{C_N}^{-1/d}<C_N^{-1/d}$, we have $\tau_N\to 0$ as $N\to \infty$.  Analogously, one can show that there is a positive sequence $\{\kappa_N\}$ satisfying (\ref {r_N}) such that  $k_N\to 0$ as $N\to \infty$ and 
$$
\lim_{N\to\infty}{S^w(v_N,\kappa_N)}=1.
$$
Then $r_N:=\min\{\tau_N,\kappa_N\}$, for $N\in \NN$,  satisfies (\ref {r_N}), $r_N\to 0$ as $N\to \infty$, and
$$
I^w(v_N,\tau_N)\leq I^w(v_N,r_N)\leq S^w(v_N,r_N)\leq S^w(v_N,\kappa_N),
$$
which implies (\ref {I}). 

For $N\in \NN$, and $x,y\in A$, we define $\mathbf{u}^0=\{u_N^0\}$ by
\begin{equation}\label{u0N}
u^0_N(x,y):=\Phi_0\(  \frac{\left|x-y\right|}{r_N}\)w(x,y),
\end{equation}
where we recall that $\Phi_0=\chi_{[0,1]}$.
It is not difficult to see that for any sequence $\{\omega_N\}$ of $N$-point configurations on $A$     we have, for $N$ sufficiently large, that
\begin {equation}\label {p5}
\begin{split}
E^{v_N}_s(\omega_N)&\geq \sum_{x,y\in \omega_N\atop 0<\left|x-y\right|\leq r_N}{\frac {v_N(x,y)}{\left|x-y\right|^s}}\geq I^w(v_N,r_N)\sum_{x,y\in \omega_N\atop 0<\left|x-y\right|\leq r_N}{\frac {w(x,y)}{\left|x-y\right|^s}}\\
&= I^w(v_N,r_N)\sum_{x,y\in \omega_N\atop x\neq y}{\frac {u^0_N(x,y)}{\left|x-y\right|^s}}\\
&= I^w(v_N,r_N)E^{u^0_N}_s(\omega_N)\geq I^w(v_N,r_N)\mathcal E^{u^0_N}_s(A,N).
\end{split}
\end {equation}
 
Let $\{\OL \omega_N\}$ be an asymptotically $\(\mathbf{v},s\)$-energy minimizing sequence of $N$-point configurations on $A$. Then by (\ref {I}) and Proposition \ref {Th1''},
\begin {equation}\label {5'}
\liminf _{N\to\infty}{\frac {E^{v_N}_s(\OL\omega_N)}{N^{1+s/d}}}\geq\lim _{N\to\infty}{\frac {I^w(v_N,r_N)\mathcal E^{u^0_N}_s(A,N)}{N^{1+s/d}}}=\frac {C_{s,d}}{\left[\mathcal H^{s,w}_d(A)\right]^{s/d}}.
\end {equation}
On the other hand,
if $\{\omega'_N\}$ is an asymptotically $(w,s)$-energy minimizing sequence of $N$-configurations on $A$, by Corollary \ref {C1}, we obtain
$$
\limsup _{N\to\infty}{\frac {E^{v_N}_s(\OL\omega_N)}{N^{1+s/d}}}=\limsup _{N\to\infty}{\frac {\mathcal E^{v_N}_s(A,N)}{N^{1+s/d}}}\leq \limsup _{N\to\infty}{\frac {E^{v_N}_s(\omega'_N)}{N^{1+s/d}}}.
$$
Then by Corollary \ref {C1}, we have
$$
E^{v_N}_s(\omega'_N)=\sum_{x,y\in \omega'_N \atop 0<\left|x-y\right|\leq r_N}{\frac {v_N(x,y)}{\left|x-y\right|^s}}+\sum_{x,y\in \omega'_N \atop \left|x-y\right|>r_N}{\frac {v_N(x,y)}{\left|x-y\right|^s}}\leq
$$
$$
\leq S^w(v_N,r_N)\sum_{x,y\in \omega'_N \atop 0<\left|x-y\right|\leq r_N}{\frac {w(x,y)}{\left|x-y\right|^s}}+C\sum_{x,y\in \omega'_N \atop \left|x-y\right|>r_N}{\frac {w(x,y)}{\left|x-y\right|^s}}\leq
$$
$$
\leq S^w(v_N,r_N)E^w_s(\omega'_N)+o(N^{1+s/d}),\ \ N\to \infty.
$$
Thus from (\ref {I}) and the fact that $\{\omega'_N\}$ is asymptotically $(w,s)$-energy minimizing, we obtain
$$
\limsup_{N\to \infty}{\frac {E^{v_N}_s(\OL\omega_N)}{N^{1+s/d}}}\leq  {\lim_{N\to \infty}{\frac { S^w(v_N,r_N)E^{w}_s(\omega'_N)}{N^{1+s/d}}}}=\frac {C_{s,d}}{\left[\mathcal H^{s,w}_d(A)\right]^{s/d}}.
$$
Taking into account (\ref {5'}), it follows that
$$
\lim_{N\to \infty}{\frac {\mathcal E^{v_N}_s(A,N)}{N^{1+s/d}}}=\lim_{N\to \infty}{\frac {E^{v_N}_s(\OL\omega_N)}{N^{1+s/d}}}=\frac {C_{s,d}}{\left[\mathcal H^{s,w}_d(A)\right]^{s/d}},
$$
which proves (\ref {asympt''}).

To prove the assertion of Theorem \ref {C2} on the limiting distribution, we use (\ref {asympt''}) and (\ref {p5}) and obtain that
$$
\frac {C_{s,d}}{\left[\mathcal H^{s,w}_d(A)\right]^{s/d}}=\lim_{N\to \infty}{\frac {E^{v_N}_s(\OL \omega_N)}{N^{1+s/d}}}\geq \limsup_{N\to \infty}{\frac  {I^w(v_N,r_N)  E^{u^0_N}_s(\OL \omega_N)}{N^{1+s/d}}}\geq 
$$
$$
\geq \liminf_{N\to \infty}{\frac  {  E^{u^0_N}_s(\OL\omega_N)}{N^{1+s/d}}}\geq 
\lim_{N\to \infty}{\frac  {\mathcal E^{u^0_N}_s(A,N)}{N^{1+s/d}}}=\frac {C_{s,d}}{\left[\mathcal H^{s,w}_d(A)\right]^{s/d}},
$$
which implies that the sequence $\{\OL \omega_N\}$ is also asymptotically $\(\mathbf{u}^0,s\)$-energy minimizing. By Proposition~\ref {Th1''}, we obtain that the sequence $\{\OL \omega_N\}$ is asymptotically uniformly distributed with respect to the measure $\mathcal H^{s,w}_d$, which completes the proof of Theorem \ref {C2}.
\end{proof}

We next provide the proof of Theorem~\ref{Th1}.

\begin{proof}[Proof of Theorem \ref {Th1}.] With $v_N$ defined as in  \eqref {v}, the boundedness of the function $\Phi$ implies that \eqref {1"} holds. 
We next verify that condition \eqref{1'} is also satisfied.  Let $a$ be a positive constant and assume $N$ is sufficiently large. If   $(x,y)\in A\times A$ is such that $0<\left|x-y\right|\leq aN^{-1/d}$, then
$$
\frac {v_N(x,y)}{w(x,y)}= \Phi\(\frac {\left|x-y\right|}{r_N}\)\geq \OL \Phi\(\frac {a}{N^{1/d}r_N}\),
$$
where the function $\OL \Phi$ is defined in \eqref {OLPhi}.
Hence,
\begin {equation}\label {qlow}
I^w\(v_N,aN^{-1/d}\)\geq \OL \Phi\(\frac {a}{N^{1/d}r_N}\),
\end {equation}
for every $N$ sufficiently large.
On the other hand, with
$$
\W \Phi(t):=\sup\limits_{u\in (0,t]}{\Phi(u)},\ \ \ t>0,
$$
we have for $(x,y)\in A\times A$, $0<\left|x-y\right|\leq aN^{-1/d}$ that
$$
\frac {v_N(x,y)}{w(x,y)}=\Phi\(\frac {\left|x-y\right|}{r_N}\)\leq \W \Phi\(\frac {a}{N^{1/d}r_N}\).
$$
Consequently, for every $N$ sufficiently large, we have
\begin {equation}\label {q1low}
S^w\(v_N,aN^{-1/d}\)\leq \W \Phi\(\frac {a}{N^{1/d}r_N}\).
\end {equation}
Since $\lim_{t\to 0^+}\OL \Phi(t)=\lim_{t\to 0^+}\W \Phi(t)=1$, letting $N\to \infty$ in \eqref {qlow} and \eqref {q1low}, we obtain condition \eqref {1'}. Then applying Theorem \ref {C2} we obtain Theorem \ref {Th1}.
\end{proof}

\section {Proof of Theorems \ref{S1} and \ref{S2}}\label{sepsec}

Throughout this section we shall assume that $A\subset \RR^{p}$ is a compact  set with $ \mathcal H_d (A)>0$. 
We first note that   Frostman's lemma  (cf. \cite[Theorem 8.8] {MatGSMES}) implies that there is a Borel measure $\mu$ on $\RR^{p}$ with  support contained in $A$ such that $0<\mu(A)<\infty$ and
\begin{equation}\label{uppermu}
\mu(B(x,r))\leq r^d,\ \ \ x\in \RR^p,\ \ \ r>0.
\end{equation}

 The proof of Theorem~\ref{S1} follows arguments from \cite{BorHarSaf08}, which in turn, use a technique from \cite {KuiSaf98}.  Also see \cite{HarSafWhi12}.   We shall appeal to the following lemma whose proof follows standard arguments as in \cite {KuiSaf98}.  

\begin {lemma}\label {sep}
Let $\omega=\{x_1,\ldots,x_N\}$ be a point configuration on  $A$ with $\mu$   satisfying \eqref{uppermu}, $r_0:=(\mu(A)/(2N))^{1/d}$,  
$$
D_i:=A\setminus \bigcup _{j:j\neq i}{B(x_j,r_0)},\ \ \ i=1,\ldots,N, 
$$
and 
$$
U_i(\omega,x):=\sum_{j:j\neq i}{\frac {1}{\left|x-x_j\right|^s}},\ \ \ x\notin \omega\setminus \{x_i\},\ \ \ i=1,\ldots,N.
$$

Then for any $s>d$ and $N\in \NN$,
$$
\frac {1}{\mu(D_i)}\int_{D_i}{U_i(\omega,x)}{\rm d}\mu(x)\leq \frac {s}{(s-d)}\(\frac {2N}{\mu(A)}\)^{s/d},\ \ \ i=1,\ldots,N.
$$
\end {lemma}

\begin{proof}[Proof of Theorem \ref {S1}.] Denote by $\W x_1,\ldots,\W x_N$ the points in the $(v_N,s)$-energy minimizing configurations $\omega_N^s$ and let
$$
U_{i,N}(x):=\sum_{j:j\neq i}{\frac {v_N(\W x_j,x)}{\left|x-\W x_j\right|^s}},\ \ \ x\in A,\ \  \ i=1,\ldots,N.
$$
Let $M>0$ be a number such that $v_N(x,y)\leq M$, $x,y\in A$, $N\in\NN$. Then for every $i=1,\ldots,N$, since $\omega_N^s$ is energy minimizing, we have
$$
E^{v_N}_s(\omega_N^s\setminus \{\W x_i\})+2U_{i,N}(\W x_i)=E^{v_N}_s(\omega_N^s)\leq E^{v_N}_s((\omega_N^s\setminus \{\W x_i\})\cup \{x\})=
$$
$$
=E^{v_N}_s(\omega_N^s\setminus \{\W x_i\})+2U_{i,N}(x)\ \ x\in A, \ \ x\neq x_1,\ldots,x_N.
$$
Hence,
$$
U_{i,N}(\W x_i)\leq U_{i,N}(x)=\sum_{j:j\neq i}{\frac {v_N(\W x_j,x)}{\left|\W x_j-x\right|^s}}
$$
$$
\leq \sum_{j:j\neq i}{\frac {M}{\left|\W x_j-x\right|^s} },\ \ \ x\in A, \ \ \ x\neq x_1,\ldots,x_N.
$$
By Lemma \ref {sep},   for   $i=1,\ldots,N$, we have  
\begin{equation}\label{UiN}
U_{i,N}(\W x_i)\leq \frac {M}{\mu(D_i)}\int_{D_i}{U_i(\omega^s_N,x)}{\rm d}\mu(x)\leq  \left(\frac {sM}{ s-d }\right)\(\frac {2N}{\mu(A)}\)^{s/d}.
\end{equation}
Clearly, it is sufficient to only consider $N$ such that $\delta(\omega^s_N)< a_0 N^{-1/d}$.  For such $N$, let $\W x_p,\W x_q\in \omega^s_N$
satisfy  $\left|\W x_p-\W x_q\right|=\delta(\omega^s_N)$. Then for every $N$ sufficiently large, using \eqref{S1cond} and \eqref{UiN}, we obtain 
\begin{equation*}
 \frac {\alpha_0}{\delta(\omega^s_N)^s} \leq   \frac {v_N(\W x_p,\W x_q)}{ \left|\W x_p-\W x_q\right|^s}\leq
  \sum_{j:j\neq p}\frac {v_N(\W x_p,\W x_j)}{\left|\W x_p-\W x_j\right|^s} 
 =  U_{p,N}(\W x_p) 
 \leq  \left(\frac{sM}{s-d}\right) \(\frac {2N}{\mu(A)}\)^{s/d},
\end{equation*}
which implies the result. \end{proof}

\begin{proof}[Proof of Theorem~\ref{S2}]
We shall adapt an argument given in \cite{HarSafWhi12}.  
Let $\omega_N^s=\{x_1,\ldots,x_N\}$ be an $N$-point $(v_N,s)$-energy minimizing  configuration for  the compact set $A$  and, for $y\in A$,   consider the  function
\begin{equation}\label{point}U(y):=\frac{1}{N}\sum_{i=1}^N \frac{v_N(y,x_i)}{|y-x_i|^s}.\end{equation}
For fixed $1\leq j\leq N$, the function $U(y)$ can be decomposed as
\begin{equation}\label{decomp}U(y)=\frac{1}{N}\frac{v_N(y,x_j)}{|y-x_j|^s}+\frac{1}{N}\sum_{\substack{i=1\\i\not=j}}^N\frac{v_N(y,x_i)}{|y-x_i|^ s},\end{equation}
and, since $\omega_N^s$ is a minimizing configuration on $A$, the point $x_j$  minimizes the sum over $i\not=j$ on the right-hand side of equation~\eqref{decomp}. Thus for each fixed $j$ and $y\in A$
\begin{align}\label{lower}
U(y)&\geq \frac{1}{N}\frac{v_N(y,x_j)}{|y-x_j|^ s}+\frac{1}{N}\sum_{\substack{i=1\\i\not=j}}^N\frac{v_N(x_j,x_i)}{|x_j-x_i|^ s}.
\end{align}
Summing over $j$ gives
 \begin{align}
 N U(y)&\geq\frac{1}{N}\sum_{j=1}^N\frac{v_N(y,x_j)}{|y-x_j|^ s}+\frac{1}{N}\sum_{j=1}^N \sum_{\substack{i=1\\i\not=j}}^N\frac{v_N(x_j,x_i)}{|x_j-x_i|^ s}
\\&= U(y)+\frac{1}{N}{\mathcal {E}}_{s}^{v_N}(A,N),
\end{align}
and thus
\begin{equation}
\label{Ulowbnd} U(y)\geq\frac{1}{N(N-1)}{\mathcal {E}}_{s}^{v_N}(A,N)\ge \frac{{\mathcal {E}}_{s}^{v_N}(A,N)}{N^2} \qquad (y\in A).\end{equation}
Since $A$ is compact, there exists a point $y^*\in A$ such that
\begin{equation}\label{y}\min_{1\leq i\leq N}|y^*-x_i|=\rho(\omega_N^s,A)=:\rho(\omega_N^s).
\end{equation}

Then, by  \eqref{asympt''} in Theorem~\ref{C2}, there is a constant $C_1>0$ and some positive integer $N_0$ such that
 \begin{equation}\label{EKlowbnd}{\mathcal {E}}_{s}^{v_N}(A,N)\ge    C_1  N^{1+s/d}\qquad (N\ge N_0).
\end{equation}
Since~\eqref{Ulowbnd} holds for the point $y^*$ of~\eqref{y}, we combine \eqref{Ulowbnd} with~\eqref{EKlowbnd} to obtain
\begin{equation}\label{lower-pot}U(y^*)\geq  \frac{{\mathcal {E}}_{s}^{v_N}(A,N)}{N^2} \ge C_1N^{s/d-1} \qquad (N\ge N_0).
\end{equation}
In addition, by equation \eqref{S1eq1} of Theorem~\ref{S1},    there is some $C_2>0$ such that $ \delta(\omega_N^s)\ge C_2N^{-1/d}$ for   $N\ge 2$.
  
Let $\mathcal{N}$ consist of those $N\ge N_0$ such that
\begin{equation}\label{reverse}\rho(\omega_N^s)\geq \frac{C_2}{2}N^{-1/d}.
\end{equation}
If $\mathcal{N}$ is empty (or finite) then we are done.   Assuming that $\mathcal{N}$ is infinite,
let $N\in \mathcal{N}$ be fixed.
 For $0<\epsilon<1/2$, let
\begin{equation}\label{r0def}
r_0=r_0(N,\epsilon):=\epsilon\,  {C_2}N^{-1/d}.
\end{equation} Note that any two of the relative balls $\tilde{B}(x_i,r_0):=  \tilde{A}\cap B(x_i,r_0)$, for $1\leq i\leq N,$ do not intersect since $r_0<{\delta(\omega_N^s)}/{2}$.
For any $x\in \tilde B(x_i,r_0)$, inequalities~\eqref{y} and \eqref{reverse} imply
\begin{align}\label{gooba}
\begin{split}
|x-y^*|&\leq |x-x_i|+|x_i-y^*|\leq r_0+|x_i-y^*|\\ &\leq  2\epsilon\, \rho(\omega_N^s) +|x_i-y^*|\leq (1+ 2\epsilon)|x_i-y^*|.\end{split}
\end{align}

Now let $\mu$ denote a  $d$-regular measure on $\tilde A$ satisfying \eqref{mureg} with positive constants $c_0, C_0$.
For fixed $1\leq i\leq N$, using ~\eqref{gooba} and  taking an average value on $\tilde{B}(x_i,r_0)$ we obtain
\begin{align}\label{mort}
\begin{split}\frac{v_N(x_i,y^*)}{|x_i-y^*|^s}&\leq \frac{C_3(1+ 2\epsilon)^s}{ \mu(\tilde{B}(x_i,r_0))}
\int_{\tilde{B}(x_i,r_0)}\frac{d\mu(x)}{|x-y^*|^s}\\ &\leq  \frac{C_3\,(1+ 2\epsilon)^s\,  c_0}{ \,r_0^{d}}
\int_{\tilde{B}(x_i,r_0)}\frac{d\mu(x)}{|x-y^*|^s},\end{split}
\end{align}
where $C_3$ denotes the uniform bound of the ${v_N}$ on $A\times A$.

Inequality  \eqref{reverse} and definition \eqref{r0def} imply $2\epsilon \rho(\omega_N^s)\ge r_0$  and thus,
for $x\in \tilde{B}(x_i,r_0)$, we obtain
\begin{align}\label{gabba}
\begin{split}
|x-y^*| &\geq|x_i-y^*|-|x-x_i| \geq |x_i-y^*|-r_0\\ &\geq  |x_i-y^*|-  2\epsilon \, \rho(\omega_N^s)\geq (1- 2\epsilon)\rho(\omega_N^s) . \end{split}
\end{align}
Inequality~\eqref{gabba} implies $$\bigcup_{i=1}^N\tilde{B}(x_i, r_0)\subset  \tilde{A}\setminus \tilde{B}(y^*,(1- 2\epsilon)\rho(\omega_N^s)),$$and since the left-hand side is a disjoint union, averaging the inequalities of~\eqref{mort}   we
have
\begin{align}\label{nadamas}
\begin{split}U(y^*)&\leq \frac{ C_3\,(1+ 2\epsilon)^s\,  c_0}{ N\,r_0^{d}}
\sum_{i=1}^N\int_{\tilde{B}(x_i,r_0)}\frac{d\mu(x)}{|x-y^*|^s}\\
&\leq\frac{ C_3\,(1+ 2\epsilon)^s\,  c_0}{ N\,r_0^{d}}\int_{\tilde{A}\setminus \tilde{B}(y^*, (1- 2\epsilon)\rho(\omega_N^s))}\frac{d\mu(x)}{|x-y^*|^s}. \end{split}\end{align}
%


Next we use the standard conversion of the integral with respect to $\mu$ to an integral with respect to Lebesgue measure (see e.g. \cite[Theorem 1.15]{MatGSMES}) to obtain
\begin{align}\label{needed}\begin{split}
 \int_{\tilde{A}\setminus \tilde{B}(y^*, (1- 2\epsilon)\rho(\omega_N^s))}\frac{d\mu(x)}{|x-y^*|^s}
  &=  \int_0^\infty\mu\{x\in\tilde{A}  : t<\frac{1}{|x-y^*|^s }<\frac{1}{[(1- 2\epsilon)\rho(\omega_N^s)]^{s}}\}\, dt\\
 & \leq  
 C_0\int_0^{((1- 2\epsilon)\rho(\omega_N^s))^{-s}}t^{-d/s}dt  \\ &=\frac{ {C}_0}{(1-d/s) (1-2\epsilon)^{s-d}} \rho(\omega_N^s)^{d-s}.
\end{split}\end{align}

Let  $N\in \mathcal{N}$. Relations \eqref{r0def}, \eqref{nadamas} and \eqref{needed}  imply
\begin{equation}\label{Uyup}
U(y^*)\le \left(\frac{ C_0C_3\,(1+ 2\epsilon)^s\,  c_0}{ (1-d/s) (1-2\epsilon)^{s-d} \epsilon^{d}C_2^d} \right) \   \rho(\omega_N^s)^{d-s}.
\end{equation}
Choosing $\epsilon=({2(2(s/d)-1)})^{-1}<\frac{1}{2}$ minimizes the right hand side of inequality \eqref{Uyup} for $\epsilon$ in $(0,1/2)$ giving
\begin{equation}\label{Uyup2}
U(y^*)\le \left[\frac{ 4^dC_0C_3  c_0s^d}{ (1-d/s)^{s-d+1}    (dC_2)^d} \right] \   \rho(\omega_N^s)^{d-s}. 
\end{equation}
Using  \eqref{lower-pot} and \eqref{Uyup2}, we then obtain
\begin{equation}\label{rhobnd1}
\rho(\omega_N^s)\le \left[\frac{ 4^dC_0C_3  c_0s^d}{ (1-d/s)^{s-d+1}   C_1 (dC_2)^d} \right]^{1/(s-d)} N^{-1/d}, 
\end{equation}
for any $N\in\mathcal{N}$.
If  $N\ge N_0$ is not in $\mathcal{N}$, then $\rho(\omega_N^s)<\frac{C_2}{2} N^{-1/d}$ and thus \eqref{limsuprho} holds. \end{proof}

\section {Proof of statements from Section \ref {section:Comp} \label{CompProofs}}

\begin{proof}[Proof of Proposition \ref {P2'}.] Denote $a:=\delta(\omega_N)/2$. For any distinct points $y_1,y_2\in \omega_N$, we have $B(y_1,a)\cap B(y_2,a)=\emptyset$.  Let $\mu$ be a $d$-regular measure on $A$ satisfying \eqref {mureg} with constants $c_0$ and $C_0$, for every point $x\in \omega_N$.  Then we have
$$
\# (\omega_N\cap B(x,\delta_N))\cdot c_0^{-1}a^d\leq \sum\limits_{y\in \omega_N\cap B(x,\delta_N)}{\mu\(A\cap B(y,a)\)}
$$
$$
=\mu\(\bigcup\limits _{y\in \omega_N\cap B(x,\delta_N)}A\cap B(y,a)\)
$$
$$
\leq \mu\(A\cap B\(x,\delta_N+a\)\)\leq C_0 (\delta_N+a)^d.
$$
Taking into account relation \eqref {well} and the fact that $\delta _N N^{1/d}=C_N$, we have
$$
Z(\omega_N,\delta_N)\leq N\max\limits_{x\in \omega_N}{\# (\omega_N\cap B(x,\delta_N))}
$$
$$
\leq  {C_0}{c_0}N \(\frac {2\delta_N}{\delta (\omega_N)}+1\)^d=O(NC_N^d),\ \ \ N\to\infty,
$$
which completes the proof of Proposition \ref {P2'}.
\end{proof}

\begin{proof}[Proof of Proposition \ref {P5}.] From the estimate
$$
E_s(\omega)\geq \sum_{x\in \omega}{\sum_{y\in \omega\atop 0<\left|y-x\right|\leq \delta}{\frac {1}{\left|y-x\right|^s}}}
$$
$$
\geq  \sum_{x\in \omega}\delta ^{-s}\# \{y\in \omega : 0<\left|y-x\right|\leq \delta\}=\delta^{-s}(Z(\omega,\delta)).
$$
Then we have $Z(\omega,\delta)\leq \delta^s E_s(\omega)$.
To prove the second part of the proposition, we write
$$
Z(\omega_N,r_N)\leq C_N^s \frac {E_s(\omega_N)}{N^{s/d}}=O(N C_N^s),\ \ \ N\to \infty. 
$$
Proposition \ref {P5} is proved.

\end{proof}

\bigskip

\noindent
{\bf Acknowledgements:}  We thank Grady Wright for generating the image in Figure 3 based on data provided by the authors using the algorithms described above.

\appendix 
\section{Proof of Proposition \ref {P3}.} \label{P3proof}

Denote
$$
{\rm dist}((x,y),(x_0,y_0))=\sqrt{\left|x-x_0\right|^2+\left|y-y_0\right|^2}
$$
and let
$$
Q_\delta (x_0)=\{(x,y)\in A\times A\setminus D(A) : 0<{\rm dist}((x,y),(x_0,x_0))\leq \delta\}.
$$
Since $w$ is a CPD weight, there is a number $\kappa>0$ such that $w(x,y)>0$, $\left|x-y\right|<\kappa$. The inequality ${\rm dist}((x,y),(x_0,x_0))\leq \delta$ implies that $\left|x-y\right|<2\delta$. Assume first that $I^w(v_N,\alpha_N)>0$. Taking into account (\ref {alpha_N}), we will have
$$
L(v_N,x_0)=\lim_{\delta\to 0}{\sup_{(x,y)\in Q_\delta (x_0)}{v_N(x,y)}}
$$
$$
\geq \lim_{\delta\to 0}{\sup_{(x,y)\in Q_\delta (x_0)}{I^w(v_N,2\delta)w(x,y)}}
$$
$$
\geq\liminf_{\delta\to 0}I^w(v_N,2\delta)\cdot  \lim_{\delta\to 0}{\sup_{(x,y)\in Q_\delta (x_0)}{w(x,y)}}
$$
$$
\geq I^w(v_N,\alpha_N)L(w,x_0).
$$
For every $x_0$ such that $L(w,x_0)=\infty$, the above estimate implies that $L(v_N,x_0)=\infty$, $N>N_0$, where $N_0$ does not depend on $x_0$. By the agreement, $L(v_N,x_0)/L(w,x_0)=1$ for $N>N_0$ and every such $x_0\in A$. Assuming now that $L(w,x_0)<\infty$, we similarly obtain
$$
L(v_N,x_0)=\lim_{\delta\to 0}{\sup_{(x,y)\in Q_\delta (x_0)}{v_N(x,y)}}
$$
$$
\leq \lim_{\delta\to 0}{\sup_{(x,y)\in Q_\delta (x_0)}{S^w(v_N,2\delta)w(x,y)}}
$$
$$
\leq\limsup_{\delta\to 0}S^w(v_N,2\delta)\cdot  \lim_{\delta\to 0}{\sup_{(x,y)\in Q_\delta (x_0)}{w(x,y)}}
$$
$$
\leq S^w(v_N,\alpha_N)L(w,x_0).
$$
Property (b) of a CPD-weight also implies that $L(w,x_0)\geq \kappa>0$. Consequently,
$$
I^w(v_N,\alpha_N)\leq \frac {L(v_N,x_0)}{L(w,x_0)}\leq S^w(v_N,\alpha_N)
$$
for every $x_0\in A$ with $L(w,x_0)<\infty$. Since quantities $I^w(v_N,\alpha_N)$ and $S^w(v_N,\alpha_N)$ do not depend on $x_0$, taking into account \eqref {alpha_N} and the fact that $L(v_N,x_0)/L(w,x_0)=1$, $N>N_0$, if $L(w,x_0)=\infty$, we obtain uniform convergence in \eqref {UCv_n}.
Using an analogous argument one can establish the second equality in \eqref {UCv_n}. 

Assume now that the weight $w$ is continuous on $D(A)$ and \eqref{alpha_N}  holds. Then $w$ is bounded by two positive constants on $D(A)$. Since $L(w,x_0)=l(w,x_0)=w(x_0,x_0)$, relations \eqref {UCv_n} imply relations \eqref {x_0}.

To establish the converse statement, 
notice that for every $m\in \NN$, there is a positive integer $N_m$ such that for $N>N_m$ and $x_0\in A$,
$$
\left|L(v_N,x_0)-w(x_0,x_0)\right|\leq \frac 1m \ \ \ {\rm and}\ \ \ \left|l(v_N,x_0)-w(x_0,x_0)\right|\leq \frac 1m.
$$
Moreover, the sequence $\{N_m\}$ can be chosen to be increasing. Then
\begin{align*}
&\lim\limits_{\delta\to 0}{\sup\limits_{(x,y)\in Q_\delta (x_0)}{\(v_N(x,y)-w(x,y)\)}}
\\ &\leq \lim\limits_{\delta\to 0}{\sup\limits_{(x,y)\in Q_\delta (x_0)}{v_N(x,y)}}
-\lim\limits_{\delta\to 0}{\inf\limits_{(x,y)\in Q_\delta (x_0)}{w(x,y)}}\\
&=L(v_N,x_0)-w(x_0,x_0)\le \frac {1}{m},\ \ x_0\in A,\ \  N>N_m.
\end{align*}
In view of property (b) of a CPD-weight, there are positive numbers $h$ and $\kappa$ such that $w(x,y)>h$ for $\left|x-y\right|<\kappa$. Let $\beta^m_N=\beta^m_N(x_0)<\kappa/\sqrt{2}$ be such that 
$$
\sup\limits_{(x,y)\in Q_{\beta^m_N} (x_0)}{\(v_N(x,y)-w(x,y)\)}<\frac 2m.
$$
The collection of open balls $\{B((x,x),\beta^m_N(x))\}_{x\in A}$ has a finite subcollection $\{B((x_i,x_i),\beta^m_N(x_i))\}$ whose union $U_{m,N}$ covers the compact set $D(A)$. Let $\alpha^m_N:={\rm dist}(D(A),(A\times A)\setminus U_{m,N})>0$. Since 
$$
Q_{m,N}:=\{(x,y)\in A\times A : 0<\left|x-y\right|\le \alpha^m_N\}\subset U_{m,N},
$$
for every $N>N_m$, we have 
$$
\sup\limits_{(x,y)\in Q_{m,N}}{\(v_N(x,y)-w(x,y)\)}
$$
$$
\leq \max \limits_ i \sup\limits _{(x,y)\in Q_{\beta^m_N}(x_i)}{\(v_N(x,y)-w(x,y)\)}\leq \frac 2m.
$$
Consequently, since $|x-y|<\kappa$ for any $x,y\in Q_{m,N}$, we have
$$
S^w(v_N,\frac{\alpha^m_N}{2})=\sup\limits_{(x,y)\in Q_{m,N}}{\frac {v_N(x,y)}{w(x,y)}}\leq 1+\frac {2}{h m},\ \ \ N>N_m.
$$
Letting $\gamma_N:=\alpha^m_N/2$, $N_m<N\leq N_{m+1}$, $m\in \NN$, we will have 
$$
\limsup\limits_{N\to \infty}{S^w(v_N,\gamma_N)}\leq 1.
$$
Using the second equality in \eqref {x_0}, one can obtain that $$
\liminf\limits_{N\to \infty}{I^w(v_N,\gamma^\prime_N)}\geq 1, 
$$
where $\{\gamma^\prime _N\}$ is some positive and bounded sequence. Then for any positive sequence $\alpha_N=o(\min \{\gamma_N,\gamma^\prime_N\})$, relation \eqref {alpha_N} still holds which completes the proof of Proposition \ref {P3}.


\begin{thebibliography}{99}
\bibitem {BenFic03}
J. Benedetto, M. Fickus, Finite normalized tight frames, {\it Adv.
Comput. Math.} {\bf 18} (2003), 357--385.
\bibitem{BenStaWil77} J. L. Bentley, D. F. Stanat, and E. H. Williams, The complexity of finding fixed-radius near neighbors,  {\it Inform. Proc. Lett.}, {\bf 6}(6) (1977), 209--212.
\bibitem {BorHarSaf07}
S.V. Borodachov, D. P. Hardin, E.B. Saff, Asymptotics of best-packing on rectifiable sets, {\it Proc. Amer. Math. Soc.}, {\bf 135} (2007), 2369--2380.
\bibitem {BorHarSaf08}
S.V. Borodachov, D.P. Hardin, E.B. Saff, Asymptotics for discrete weighted minimal Riesz energy problems on rectifiable sets, {\it Trans. Amer. Math. Soc.}, {\bf 360} (2008), 1559-1580.
\bibitem{Bow00} M. Bowick,    D.R. Nelson,  A.  Travesset, Interacting topological defects in 
frozen topographies, {\it Phys. Rev. B} {\bf 62}  (2000),  8738--8751.
\bibitem {BrauHarSaff}
J.S. Brauchart, D.P. Hardin, E.B. Saff, The next-order term for optimal Riesz and logarithmic energy asymptotics on the sphere, {\it Contemp. Math.}, {\bf 578} (2012), 31--61.
\bibitem{Cha83}  B. Chazelle, An improved algorithm for the fixed-radius neighbor problem,  {\it Inform. Proc. Lett.},  {\bf 16} (1983), 193--198. \bibitem {ConSlo99}
J.H. Conway, N.J.A. Sloane, {\it Sphere Packings, Lattices and Groups,}
Springer Verlag, New York: 3rd ed., 1999.
\bibitem {FedGMT}
H. Federer, {\it Geometric Measure Theory}, Springer-Verlag, 1969.
\bibitem {HarSaf05}
D.P. Hardin, E.B. Saff, Minimal Riesz energy point configurations for rectifiable $d$-dimensional manifolds, {\it Adv.  Math.}, {\bf 193} (2005), no. 1, 174--204.
\bibitem{HarSafWhi12} D. P. Hardin, E. B. Saff, and J. T. Whitehouse, Quasi-uniformity of minimal weighted energy points, {\it J.  Complexity} {\bf 28} (2012), 177-191.
\bibitem {KuiSaf98}
A.B.J. Kuijlaars, E.B. Saff, Asymptotics for minimal discrete energy on the sphere, {\it Trans. Amer. Math. Soc.} {\bf 350} (1998), no. 2, 523--538.
\bibitem {MMRS}
A. Martinez-Finkelshtein, V. Maymeskul, E. A. Rakhmanov, and E.B. Saff, Asymptotics for minimal discrete Riesz energy on curves in $R^d$, {\it Canad. J. Math.}, {\bf 56} (2004), 529--552.
\bibitem {MatGSMES}
P. Mattila, {\it Geometry of Sets and Measures in Euclidean Space}, Cambridge University Press, 1995.
\bibitem{Mel77} T.W. Melnyk,  O. Knop, and W. R. Smith,  Extremal arrangements of points and unit charges on a sphere: equilibrium configurations revisited,  {\it Canad. J. Chem.} {\bf 55}, (1977), 1745--1761.
\bibitem{SloWom} I.H. Sloan, R.S. Womersley, Extremal systems of points and numerical integration on the sphere, {\it Adv. Comp. Math.} {\bf 21} (2004), 102--125. 
\end {thebibliography}
\end {document}